\documentclass{article}

\usepackage{amsmath,amssymb,amsthm,tikz,caption,subcaption,mathtools}
\usepackage[ruled,vlined,noresetcount]{algorithm2e}
\usepackage{amsfonts}
\usepackage{microtype}
\usepackage{xcolor}
\usepackage{hyperref}
\hypersetup{colorlinks=true,linkcolor=red,citecolor=magenta,urlcolor=blue}
\usepackage{cleveref}
\usepackage{bm}
\usepackage{bbm}
\usepackage{float}
\usepackage{enumitem}
\setlist{
  listparindent=\parindent,
}
\usepackage[margin=1in]{geometry}
\usepackage{tikz-cd}

\newtheorem{theorem}{Theorem}[section]
\newtheorem{corollary}[theorem]{Corollary}

\newtheorem{lemma}[theorem]{Lemma}
\newtheorem{definition}[theorem]{Definition}

\newtheorem{remark}[theorem]{Remark}

\newtheorem{claim}[theorem]{Claim}

\newcommand{\RR}{\mathbb{R}}

\newcommand{\ZZ}{\mathbb{Z}}

\newcommand{\cE}{\mathcal{E}}
\newcommand{\cF}{\mathcal{F}}
\newcommand{\cG}{\mathcal{G}}
\newcommand{\cS}{\mathcal{S}}
\newcommand{\cT}{\mathcal{T}}
\newcommand{\cU}{\mathcal{U}}

\newcommand{\set}[1]{\{ #1 \}}

\DeclarePairedDelimiter{\abs}{\lvert}{\rvert}
\DeclarePairedDelimiter{\norm}{\lVert}{\rVert}

\newcommand{\wt}{\widetilde}
\newcommand{\wh}{\widehat}

\newcommand{\interior}[1]{%
  {\kern0pt#1}^{\mathrm{o}}%
}

\newcommand{\<}{\langle}
\renewcommand{\>}{\rangle}

\DeclareMathOperator{\Ker}{ker}
\DeclareMathOperator{\Ima}{im}

\let\im\relax

\DeclareMathOperator{\im}{im}
\DeclareMathOperator{\supp}{supp}

\DeclareMathOperator{\polylog}{polylog}
\DeclareMathOperator{\Set}{Set}

\newcommand{\code}{\textnormal{code}}
\newcommand{\sq}{\textnormal{square}}
\newcommand{\euclid}{\textnormal{euclid}}

\newcommand{\sur}{\textnormal{sur}}
\newcommand{\rep}{\textnormal{rep}}
\newcommand{\qLDPC}{\textnormal{qLDPC}}
\newcommand{\LTC}{\textnormal{LTC}}
\newcommand{\inter}{\textnormal{int}}
\newcommand{\bound}{\partial}
\renewcommand{\int}{\textnormal{int}}
\newcommand{\bnd}{\partial}

\def\F2{\mathbb{F}_2}

\title{Optimal Geometrically Local Quantum and Classical Codes \\ from Subdivision}
\author{Ting-Chun Lin\thanks{Department of Physics, University of California San Diego, CA, and Foxconn Research, Taipei, Taiwan. Email: \texttt{til022@ucsd.edu}.} \and Adam Wills\thanks{ Foxconn Research, Taipei, Taiwan. Email: \texttt{adamjwills7248@gmail.com}.}\and Min-Hsiu Hsieh\thanks{Foxconn Research, Taipei, Taiwan. Email: \texttt{min-hsiu.hsieh@foxconn.com}.}}

\usepackage{booktabs}

\begin{document}

\sloppy

\maketitle

\begin{abstract}
  A geometrically local quantum code is an error correcting code situated within $\RR^D$, where the checks only act on qubits within a fixed spatial distance.
  The main question is: What is the optimal dimension and distance for a geometrically local code?
  Recently, Portnoy made a significant breakthrough with codes achieving optimal dimension and distance up to polylogs.
  However, the construction invokes a somewhat advanced mathematical result that involves lifting a chain complex to a manifold.
  This paper bypasses this step and streamlines the construction by noticing that a family of good quantum low-density parity-check codes, balanced product codes, naturally carries a two-dimensional structure.
  Together with a new embedding result that will be shown elsewhere,
    this quantum code achieves the optimal dimension and distance in all dimensions.
  In addition, we show that the code has an optimal energy barrier.
  We also discuss similar results for classical codes.
\end{abstract}

\section{Introduction}
Coding theory has been the center of many applications and is especially important in quantum computing to maintain quantum coherence under noise.
One particular focus is on quantum codes whose parity-checks only act on a few qubits,
  also known as the quantum low-density parity-check (qLDPC) codes.
The LDPC property is favorable for quantum codes
  because quantum information is highly sensitive to noise
  and measurements often cause errors.
The fewer the qubits involved, the lower the measurement error.
This is in sharp contrast to practical classical codes
  which suffer from essentially no measurement error.
Finding qLDPC codes with large distance and dimension has been an important open question in coding theory; it is only recently that such codes have been constructed \cite{panteleev2021asymptotically}.

However, for certain applications, having low density is not enough.
Since we live in a three dimensional world, the code should be embedded in $\RR^3$.
We want each check to be implementable via short range measurements
and we want the qubit density to be finite.
Such codes are called geometrically local.

Finding geometrically local codes with large distance and dimension has also been an important question for physicists and coding theorists.
Back in 2008, Bravyi and Terhal \cite{bravyi2009no} showed an upper bound on code distance and in 2009, Bravyi, Poulin, and Terhal \cite{bravyi2010tradeoffs} extended the previous result to an upper bound on the tradeoff between code dimension and code distance.
At the time, the only known geometrically local codes were the high-dimensional toric codes, which saturate the upper bound in two dimensions, but not in higher dimensions. As such, it was thought that this upper bound is not tight and can be improved.
However, Haah in 2011 \cite{haah2011local} and Michnicki in 2012 \cite{michnicki20143d}
  constructed new geometrically local codes that beat the toric codes, thus improving the lower bound.
Despite these advancements, the gap between the upper and lower bounds remained open.

Recently, Portnoy \cite{portnoy2023local} finally closes the gap up to polylogs.
The idea is to take the good qLDPC codes mentioned earlier,
  ``geometrize'' the codes which turns them into manifolds based on the results of Freedman and Hastings \cite{freedman2021building},
  and finally embed the manifolds into $\RR^D$ using the work of Gromov and Guth \cite{gromov2012generalizations}.
It is remarkable how the construction takes a code to geometry and then back to code again.

However, one unsatisfactory aspect is that the step that turns codes into manifolds is somewhat complicated.
In this paper, we circumvent this step by utilizing the structure of the good qLDPC codes people have constructed in a non-black box manner.
In particular, the structure we will use is the balanced product introduced in \cite{breuckmann2021balanced}.\footnote{Similar structures have been discovered independently in many works, including the lifted-product in \cite{panteleev2021asymptotically} and the left-right Cayley graph in \cite{dinur2022good}.}
As a result, our code construction is more explicit\footnote{The term ``explicit'' here is used in a non-techincal sense. In theoretical computer science, ``explicit'' typically means constructable in polynomial time, and in this sense, both codes by Portnoy and ourselves can be made explicit.}.
To demonstrate this advantage, we show that our code additionally has the optimal energy barrier.

\begin{remark}
  Other works with related results by Baspin \cite{baspin2023combinatorial}  and Williamson-Baspin \cite{williamson2023layer} appeared independently with our work.
  While these works are about classical and quantum codes that saturate the BPT bound,
    the results were derived independently and
    the details of the constructions are different.
  We will compare the codes later in \Cref{sec:comparison}.
\end{remark}

\subsection{Main Contributions}

Our main contribution is to identify a 2D geometrical structure for the balanced product codes.
This enables us to simplify the previous construction \cite{portnoy2023local} by avoiding the step that turns codes into manifolds.
Our code gives the optimal distance and energy barrier.

\begin{theorem} \label{thm:quantum-LDPC}
  There exists a family of $D$-dimensional geometrically local quantum codes ($D \ge 3$)
  with code dimension $k = \Omega(n^\frac{D-2}{D})$,
    distance $d = \Omega(n^\frac{D-1}{D})$,
    and energy barrier $\cE = \Omega(n^\frac{D-2}{D})$.
\end{theorem}

\begin{remark} \label{remark:no-polylog}
  In the early version of this work,
    similar to Portnoy's result,
    our construction only achieves the optimal code parameters up to a polylog.
  This limitation comes from the suboptimal embedding \cite[Theorem 7]{portnoy2023local}.

  After learning about Williamson and Baspin's work where they achieved optimal distance without a loss of polylog factor,
    we realized that we can also eliminate the polylog factor by constructing an optimal explicit embedding of the square complex in $\RR^D$.
  The detailed construction of this embedding will be presented in a forthcoming paper.
\end{remark}

This answers the questions raised in Portnoy's work \cite{portnoy2023local}.
For Question 1, as discussed in \Cref{remark:no-polylog},
  it is possible to remove the polylog factors by constructing an explicit embedding of the simplicial complexes which will be discussed in a future paper.
We answer Question 2 affirmatively and show that we can bypass the manifolds (at least for some good qLDPC codes).
Moreover, our code likely has a decoder based on surface codes which potentially resolves Question 3.

Besides these results on geometrically local quantum codes,
  by applying the same idea,
  we obtain geometrically local classical codes with
    optimal distance, and energy barrier.
Additionally, this code has nontrivial soundness, which utilizes the recent results on $c^3$ locally testable codes \cite{panteleev2021asymptotically,dinur2021locally,lin2022c}.

\begin{theorem} \label{thm:classical-LTC}
  There exists a family of $D$-dimensional geometrically local classical codes ($D \ge 2$)
  with code dimension $k = \Omega(n^\frac{D-1}{D})$,
    distance $d = \Omega(n)$,
    energy barrier $\cE =\Omega(n^\frac{D-1}{D})$,
    and soundness $s = \Omega(n^{-\frac{1}{D}})$.
\end{theorem}

As a consequence of the codes above, we obtain codes that saturate the tradeoff between code dimension and distance.
This is achieved by copying the codes multiple times and stacking them into a lattice grid.
In particular, one can construct a new family of codes of size $L^D$
  by stacking $a^D$ many codes of size $(L/a)^D$.
By choosing a suitable value for $a$, we obtain the following code family.
\begin{corollary}
  For any $\tau \ge \frac{D-2}{D}$ ($D \ge 3$),
    there exists a family of $D$-dimensional geometrically local quantum codes
  with code dimension $k = \Omega(n^\tau)$,
    distance $d = \Omega((n/k)^\frac{D-1}{2})$,
    and energy barrier $\cE = \Omega((n/k)^\frac{D-2}{2})$.
\end{corollary}

\begin{corollary}
  For any $\tau \ge \frac{D-1}{D}$ ($D \ge 2$),
    there exists a family of $D$-dimensional geometrically local classical codes
  with code dimension $k = \Omega(n^\tau)$,
    distance $d = \Omega((n/k)^D)$,
    energy barrier $\cE =\Omega((n/k)^{D-1})$,
    and soundness $s = \Omega((n/k)^{-1})$.
\end{corollary}
As discussed in \Cref{sec:upper-bound},
  $D$-dimensional geometrically local quantum codes satisfy $d = O((n/k)^\frac{D-1}{2})$
  and the $D$-dimensional geometrically local classical codes satisfy
  $d = O((n/k)^D)$.
Therefore, the codes above saturate the tradeoff between code dimension and distance.
Whether the tradeoff between code dimension and energy barrier is optimal is left as an open question.

Furthermore, our construction method readily generalizes to other settings,
  including codes derived from cubical complexes \cite{dinur2024expansion}
  and local Hamiltonians beyond codes, although their properties require further investigation.

\subsection{Construction and Proof Overview}
One of the difficulties in constructing geometrically local codes
  is that good qLDPC codes are geometrically non-local.
This is because good qLDPC codes are based on expander graphs
  which do not have geometrically local embeddings.
Nevertheless, we can still force an embedding
  by paying the cost that the vertices in the graph are far apart.
If one continues this thought,
  a way to get a graph with a local embedding
  is to insert new vertices on each long range edge, subdividing them into constant sized edges.
This gives a graph which is geometrically local
  yet retains some of the expansion properties.
A similar idea applies to good classical LDPC codes
  whose parity-check matrix corresponds to the adjacency matrix of a bipartite graph.
This procedure gives a geometrically local classical code
  with the optimal dimension and distance.

We essentially apply the same idea to quantum LDPC codes. Many of the known good qLDPC codes can be phrased in terms of a balanced product.
The balanced product codes may be obtained by taking a certain product of two classical codes.
To embed the resulting complex in $\RR^D$, all we need to do is to subdivide each edge and face by adding vertices. The subdivision process is illustrated in \Cref{fig:intro-subdivision}.

\begin{figure}[H]
  \centering
  \includegraphics[width=0.7\textwidth]{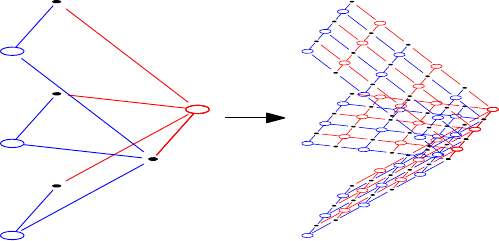}
  \caption{The subdivision of a balanced product code. (Balanced product code is described in \Cref{sec:prelim-good-qldpc}.) Blue circles, black dots and red circles represent X - stabilisers, qubits and Z - stabilisers, respectively.}
  \label{fig:intro-subdivision}
\end{figure}

The proof mainly involves properties of the chain map $\cF$, which maps from $\F2^{X(0)} \to \F2^{X(1)} \to \F2^{X(2)}$, the chain complex of the original good qLDPC code, to $\F2^{X_L(0)} \to \F2^{X_L(1)} \to \F2^{X_L(2)}$, the chain complex of the subdivided code.
\begin{equation}
  \begin{tikzcd}
    \F2^{X(0)} \arrow[r, "\delta_0"] \arrow[d, "\cF_0"] &
    \F2^{X(1)} \arrow[r, "\delta_1"] \arrow[d, "\cF_1"] &
    \F2^{X(2)} \arrow[d, "\cF_2"] \\
    \F2^{X_L(0)} \arrow[r, "\delta_0"] & \F2^{X_L(1)} \arrow[r, "\delta_1"] & \F2^{X_L(2)}
  \end{tikzcd}
\end{equation}
The strategy to show any desired property in $X_L$, say distance,
  is to reduce to the same property in $X$.
This reduction relies on codes that are similar to the surface code and the repetition code that arise naturally when we perform the subdivision.
We call these the generalized surface codes and the generalized repetition codes, as illustrated in \Cref{fig:intro-generalized-code}.
A helpful intuition is that the subdivided code is like
  a concatenated code where the inner code is the generalized surface and repetition codes
  and the outer code is the good qLDPC code.
The actual case is more nuanced.

\begin{figure}[H]
  \centering
  \includegraphics[width=0.7\linewidth]{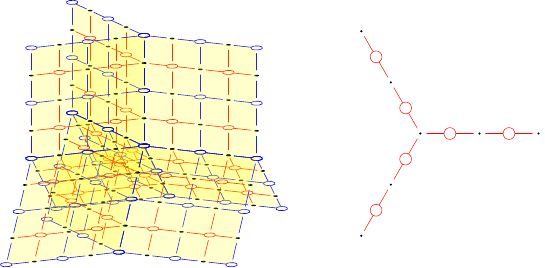}
  \vspace{1em}
  \caption{The generalized surface code and generalized repetition code.}
  \label{fig:intro-generalized-code}
\end{figure}

\subsection{Comparison} \label{sec:comparison}

We now compare the constructions of geometrically local quantum codes presented in Portnoy \cite{portnoy2023local}, Williamson-Baspin \cite{williamson2023layer}, and this paper.
All three constructions share a commonality in that they can all be interpreted as gluing 2D surface codes in a clever way based on good qLDPC codes.
The approach taken by Portnoy and ourselves involves viewing the good qLDPC codes as 2D geometrical structures
  and constructing new codes by associating each 2D structure with a surface code.
On the other hand, the approach taken by Williamson-Baspin involves replacing each qubit and check with a surface code.
However, a direct replacement leads to non-commuting checks, so it requires additional attachment of the ``defect lines''.

Both approaches have their own advantages and drawbacks, as summarized in \Cref{table:comparison}.
In particular, the 2D geometrical structures are found more easily in certain families of codes, so the first approach does not work for arbitrary codes yet.
On the other hand, the second approach currently works more naturally in three-dimension, partly due to the structure of the ``defect lines''.

\begin{table}[H]
  \centering
  \caption{Comparison of the constructions}
  \label{table:comparison}
  \begin{tabular}{ccccc}
    \toprule
    & \textbf{BPT bound?} & \textbf{Energy Barrier?} & \textbf{Arbitrary code?} & \textbf{All dimensions?} \\
    \midrule
    Portnoy & \checkmark $^\sharp$ & \checkmark $^\natural$ & Almost$^\flat$ & \checkmark \\
    Our work & \checkmark $^\sharp$ & \checkmark & Only balanced product & \checkmark \\
    Williamson-Baspin & \checkmark & \checkmark & \checkmark & Currently in 3D \\
    \bottomrule
  \end{tabular}
  \flushleft

  $^\sharp$After using the optimal embedding.
  $^\natural$Not stated, but it holds.
  $^\flat$Requires the code to have a sparse $\ZZ$-lift \cite[Definition 1.2.2]{freedman2021building}.
\end{table}

\subsection{Related Work}

\paragraph*{Good qLDPC codes}
Our result relies heavily on the recent progress of good qLDPC codes.
We summarize some recent advances.
One main open question in quantum coding theory was whether qLDPC codes with linear distance and dimension exist.
The earliest qLDPC codes are Kitaev's toric codes and surface codes \cite{kitaev2003fault} with dimension $k=\Theta(1)$ and distance $d=\Theta(\sqrt{n})$.
Progress in the field led to the discovery of codes with improved rates $k=\Theta(n)$ \cite{tillich2013quantum} and better distances $d=\Theta(\textnormal{polylog}(n) \sqrt{n})$ \cite{freedman2002z2, evra2020decodable, kaufman2020new}.
A significant breakthrough occurred when \cite{hastings2021fiber} improved the distance by a polynomial factor and broke the square root barrier.
  Following works achieved nearly linear distance $d=\Theta(n/\log n)$ \cite{panteleev2021quantum} and
  discovered the framework of balanced product and derandomized the earlier construction \cite{breuckmann2021balanced}.
Finally, \cite{panteleev2021asymptotically} showed the existence of good quantum LDPC codes with $k=\Theta(n)$ and $d=\Theta(n)$.
Subsequently, two more constructions were discovered \cite{leverrier2022quantum,dinur2022good}.
All these codes were shown to have a linear time sequential decoder \cite{gu2022efficient,leverrier2022efficient,dinur2022good}
  and a log time parallel decoder \cite{leverrier2022parallel}.
Besides coding theory, these codes also have applications in complexity theory \cite{hopkins2022explicit}, notably in resolving the NLTS conjecture \cite{anshu2022nlts}.

\paragraph*{Geometrically Local Quantum Codes}
As we live in a three-dimensional world,
  it is natural to explore codes that can be embedded in 3D,
  which have a potential realization in physical materials.
This exploration also extends to codes that can be embedded in higher dimensions.
We first discuss the lower bound of the parameters for the local quantum codes, and then we discuss the upper bound.

The earliest local quantum code studied was the generalized toric code \cite{dennis2002topological}.
The code is constructed by subdividing the $D = D_1 + D_2$ dimensional torus into grids of size $L^D$.
Subsequently, qubits are placed on $D_1$-cells,
  X-checks on $(D_1-1)$-cells,
  and Z-checks on $(D_1+1)$-cells.
Since the torus has nontrivial $D_1$-homology,
  the code is nontrivial with dimension $\Theta(1)$.
Additionally, since the nontrivial X (Z) logical operator corresponds to a $D_2$ ($D_1$) dimensional surface,
  one can show that the code has X-distance $\Theta(L^{D_2}) = \Theta(n^{\frac{D_2}{D}})$ and Z-distance $\Theta(L^{D_1}) = \Theta(n^{\frac{D_1}{D}})$
  where $n = \Theta(L^D)$ is the number of qubits.
In physics, this code is also known as the $\ZZ_2$ gauge theory which is a topological quantum field theory (TQFT).
Codes from TQFT were the best codes known at the time.
As a comparison to the codes that will be described later,
  the 3D toric code has parameters $[[n=\Theta(L^3), k=\Theta(1), d=\Theta(L), \cE = \Theta(1)]]$.

To obtain better codes, we have to go beyond TQFT.
In particular, two novel code constructions appeared.
One was by Haah in 2011 \cite{haah2011local} which constructed the first 3D gapped Hamiltonian beyond TQFT.
This code has parameters $[[n=\Theta(L^3), k=\Theta(L), d=\Omega(L^\gamma), \cE = \Omega(\log L)]]$ for some $\gamma > 1$ \cite{bravyi2011energy}.
In particular, it has superlinear distance and non constant energy barrier.
This breakthrough not only changes our understand of local quantum codes, but also sparks the nascent field of fractons in condensed matter physics.

Haah discovered the code by asking whether a 3D lattice code could exist without string operators, which is a characteristic property of TQFTs.
He observed that lattice codes in 3D with translation invariance can be compactly described as modules over the polynomial ring $\F2[x,y,z]$.
He then rephrased the no-string rule as a property in commutative algebra
  and discovered this nontrivial code.

The other construction was introduced by Michnicki in 2012 \cite{michnicki20143d}.
The new technique, known as ``welding'',
  involves welding together surface codes along the edges to form a larger code.
(In fact, our code can also be viewed as a welded code, using good qLDPC codes as a guide for the welding process.)
The code has parameters $[[n=\Theta(L^3), k=1, d=\Theta(L^{\frac{4}{3}}), \cE = \Theta(L^{\frac{2}{3}})]]$.
In particular, it has a polynomial energy barrier.
Notably, this is the first geometrically local quantum code that does not have translation invariance.

The most recent code by Portnoy in 2023 \cite{portnoy2023local}
  has parameters $[[n=\Theta(L^3), k=\Omega(L/\polylog L), d=\Omega(L^2/\polylog L)]]$.
In our work, we extend the parameters to include the energy barrier
  $[[n=\Theta(L^3), k=\Omega(L), d=\Omega(L^2), \cE=\Omega(L)]]$.

Besides the question on local quantum codes,
  one can also ask about local classical codes.
The work by Yoshida in 2011 \cite{yoshida2013information}
  constructed local classical codes based on fractals.
This achieves the optimal scaling between code dimension and distance up to polylogs
  $[n=\Theta(L^D), k=\Omega(L^{D-1}/\polylog L), d=\Omega(L^D/\polylog L), \cE=\Theta(\log L)]$.
Our work improves the energy barrier and show that the code has nontrivial soundness
  $[[n=\Theta(L^D), k=\Omega(L^{D-1}), d=\Omega(L^D), \cE=\Omega(L^{D-1}), s=\Omega(L^{-1})]]$.
We remark that quantum codes based on fractals have also been explored \cite{brell2016proposal,zhu2022topological,dua2023quantum}, but either the embedding in 3D is unknown
  or the code does not have superlinear distance.
This concludes the discussion on the lower bound.

The upper bound on distance and energy barrier for quantum and classical codes was studied in \cite{bravyi2009no}
  and the tradeoff between code dimension and distance was studied in \cite{bravyi2010tradeoffs}.
These imply that the optimal $D$ dimensional classical code has distance $\Theta(L^D)$ and energy barrier $\Theta(L^{D-1})$.
  Under these conditions the optimal dimension is $\Theta(L^{D-1})$.
Similarly, the optimal $D$ dimensional quantum code has distance $\Theta(L^{D-1})$ and energy barrier $\Theta(L^{D-2})$.
  Under these conditions the optimal dimension is $\Theta(L^{D-2})$.
Specific results can be found in \Cref{sec:upper-bound}.

On the other hand, through a different starting assumption, \cite{haah2021degeneracy} showed that
  all quantum codes that have ``homogeneous topological order''
  have upper bounds on the code dimension $k = O(L^{D-2})$.
Notice that all the codes discussed above have homogeneous topological order.

\subsection{Further Directions}

\paragraph*{Efficient Decoder?}
Can we show the code discussed in this paper has an efficient decoder that decodes correctly up to a constant times the distance?
It seems plausible this decoder could be obtained by combining the decoder for the surface code and the decoder for the original qLDPC code.

\paragraph*{``Geometrize'' Quantum Tanner Code by Leverrier and Z\'emor?}
All known constructions of good qLDPC codes can be viewed as a balanced product which
  induces a geometric embedding
  with the exception of the quantum Tanner code \cite{leverrier2022quantum}.
Can we induce a geometric embedding for the quantum Tanner code?\footnote{In the previous version of the paper, there was a construction proposed. However, it has been removed as it was found to be incorrect. A proper construction will be presented in a forthcoming work.}

\paragraph*{Simpler Manifold Construction?}
This paper bypasses Freedman, Hastings \cite{freedman2021building} by noticing balanced product codes naturally correspond to square complexes.
Can we use this feature to construct a manifold corresponding to the balanced product code in a simpler way?
Perhaps this leads to a lower dimensional manifold construction.
See \cite[Sec 1.9]{freedman2021building} for further context.

\paragraph*{Geometric Embedding Beyond Codes?}
This paper focuses on converting a qLDPC code into a geometrically local quantum code.
We wonder if this reduction can be applied to other questions.
For example, we know that the local Hamiltonian question is QMA-hard.
Can we obtain a reduction to questions about geometrically local Hamiltonians?
Or, the SYK model is given by the random 4-fermion Hamiltonian.
Can we convert this into a geometrically local Hamiltonian?
These questions can be posed in the context of both traditional classical CSP problems and quantum Hamiltonians.

\paragraph*{Geometrically Local Codes with Translation Invariance?}
This paper shows a matching upper and lower bound for distance and energy barrier of geometrically local codes.
What if we additionally impose translation invariance?
The code constructed by Haah in 3D of size $L \times L \times L$ was shown to have superlinear distance $\Omega(L^\gamma)$ for some $\gamma > 1$, and energy barrier $\Omega(\log L)$ by Bravyi and Haah \cite{bravyi2011energy}.
Neither of these match the upper bound in \cite{bravyi2009no}.
Can we possibly close this gap?

We note that, currently, all known translation invariant codes have logarithmic energy barriers.
This includes the classical code constructed by Yoshida \cite{yoshida2013information} and the quantum code constructed by Haah \cite{haah2011local}.
The origin of the logarithm is that the logical operator of these codes are Sierpiński gaskets, which have bad isoperimetric properties, resulting in a logarithmic energy barrier.
The Sierpiński gasket appears in both settings
  because of its nice property that it can be induced with translation invariant Hamiltonian
  as observed in glassy spin models \cite{newman1999glassy}.
Overcoming this logarithmic barrier represents an interesting milestone in advancing our understanding of translation invariant code.
One potential avenue is to have translation invariant codes whose codewords resemble Sierpiński carpets, which have good isoperimetric properties.

\paragraph*{Geometric Embedding in Fractal Geometries?}
This paper achieves the optimal code dimension and distance in $\ZZ^D$.
What about the optimal code in fractal geometries?
Similar questions have been explored previously in \cite{zhu2022topological,dua2023quantum}.

\paragraph*{Self-Correcting Quantum Memory?}
Self-correcting quantum memories \cite{dennis2002topological,alicki2010thermal,eczoo_self_correct} are the quantum analog of magnetic tapes.
Magnetic tapes are passive, non-volatile memories that can retain information for a long time without active error correction.
A natural question is: Can we find geometrically local quantum codes in 3D with similar properties?
It is known that 4D toric codes have this property \cite{dennis2002topological,alicki2010thermal}
  and 2D stabilizer codes do not have this property \cite{bravyi2009no}.
However, it remains open whether 3D self-correcting quantum memories exist.

Note that even though our code has large energy barrier,
  we do not think our code is self-correcting.
This is because our code utilizes the 2D toric code and probably shares similar thermal properties.

\section{Preliminary}

\subsection{Chain Complexes}
\label{sec:prelim-chain-complexes}

Chain complexes provide a natural framework for describing both quantum CSS codes and certain constructions of classical locally testable codes, which are the focus of this work.
Later, we will see that many properties of the codes, including distance, energy barrier for quantum codes, and local testability for classical codes, can be interpreted in terms of expansion properties of the chain complex defined in \Cref{sec:prelim-chain-complex-expansion}.

\begin{definition}[Chain complex]
  A chain complex $X$ consists of a sequence of vector spaces $\F2^{X(i)}$ generated by sets $X(i)$, along with linear maps $\delta_i: \F2^{X(i)} \rightarrow \F2^{X(i+1)}$ known as coboundary operators, where the coboundary operators satisfy
  \begin{equation*}
    \delta_{i+1} \delta_i = 0.
  \end{equation*}
\end{definition}

Given a chain complex, one can change the direction of the linear maps and form the dual chain complex consists of the boundary operators.
In our context, there is a canonical basis of $\F2^{X(i)}$ labeled by the elements of $X(i)$.
Using this basis, we can define the associate boundary operators $\partial_i: \F2^{X(i)} \rightarrow \F2^{X(i-1)}$ as the matrix transpose of $\delta_{i-1}$, $\partial_i \coloneqq \delta_{i-1}^T$.
The boundary operators automatically satisfy
\begin{equation*}
  \partial_{i-1} \partial_i  = 0
\end{equation*}
which is again a chain complex.

We introduce some standard terminologies.
Elements of the kernel of the (co)boundary operators are called (co)cycles
\begin{equation*}
  Z_i \coloneqq \Ker \partial_i = \{c_i \in \F2^{X(i)} : \partial_i c_i = 0\}, \qquad
  Z^i \coloneqq \Ker \delta_i = \{c_i \in \F2^{X(i)} : \delta_i c_i = 0\}.
\end{equation*}
Elements of the image of the (co)boundary operators are called (co)boundaries
\begin{equation*}
  B_i \coloneqq \Ima \partial_{i+1} = \{\partial_{i+1} c_{i+1} : c_{i+1} \in \F2^{X(i+1)}\}, \qquad
  B^i \coloneqq \Ima \delta_{i-1} = \{\delta_{i-1} c_{i-1} : c_{i-1} \in \F2^{X(i-1)}\}.
\end{equation*}
Because $\delta_i \delta_{i-1} = 0$, it follows that $B^i \subset Z^i$.
When $B^i = Z^i$ the chain complex is said to be \emph{exact} at $i$.

\subsection{Classical and Quantum Error Correcting Codes}
\label{sec:prelim-codes}

\subsubsection{Classical Error Correcting Codes}

A classical code is a $k$-dimensional linear subspace $C \subset \F2^n$
  which is specified by a parity-check matrix $H: \F2^n \to \F2^m$
  where $C = \Ker H$.
$n$ is called the size and $k$ is called the dimension.
The distance $d$ is the minimum Hamming weight of a nontrivial codeword
\begin{equation}
  d = \min_{c \in C - \{0\}} |c|.
\end{equation}

The energy barrier $E$ is the minimum energy required to generate a nontrivial codeword by flipping the bits one at a time.
More precisely, given a vector $c \in \F2^n$,
  $|H c|$ is the number of violated checks.
In the physical context, the energy of a state is proportional to the number of violated checks,
  so we will refer to $\epsilon(c) = |H c|$ as the energy of vector $c$.
We say a sequence of vectors $\gamma_{a\to b} = (c_0=a, c_1, ..., c_t=b)$ is a walk from $a$ to $b$
  if $c_i, c_{i+1}$ differ by exactly one bit $|c_i - c_{i+1}| = 1$.
The energy of a walk $\epsilon(\gamma) = \max_{c_i \in \gamma} \epsilon(c_i)$ is defined as the maximum energy reached among the vectors in $\gamma$.
Finally, the energy barrier $\cE$ is defined as the minimum energy among all walks $\gamma$ from $0$ to a nontrivial codeword
\begin{equation}
  \cE = \min_{\gamma_{0\to c}, c \in C - \{0\}} \epsilon(\gamma_{0\to c}).
\end{equation}
Intuitively, energy barrier is another way to characterize the difficulty for a logical error to occur, other than distance.

We say the classical code is a low-density parity-check (LDPC) code if each check interacts with a bounded number of bits, i.e. $H$ has a bounded number of nonzero entries in each row.

We say the classical code is (strongly) locally testable with soundness $s$ if it satisfies
\begin{equation}
  \forall x \in \F2^n : \frac{1}{m} |Hx| \ge \frac{s}{n} \min_{c \in C} |x-c|.
\end{equation}

Notice that if the code has large distance and soundness
  then the code has large energy barrier
  where $\cE \ge s \frac{m}{n} \frac{d}{2} = \Theta(sd)$.
The reason is that if $\epsilon(x) = |Hx| < s \frac{m}{n} \frac{d}{2}$,
  by soundness, $\min_{c \in C} |x-c| < d/2$.
Therefore, a walk $\gamma$ starting from $0$ with energy $< s \frac{m}{n} \frac{d}{2}$ cannot reach other codewords.
Hence, $\cE \ge s \frac{m}{n} \frac{d}{2}$.

\subsubsection{Quantum CSS Codes}

A quantum CSS code is specified by two classical codes $C_x, C_z$
  represented by their parity-check matrices $H_x: \F2^n \to \F2^{m_z}, H_z: \F2^n \to \F2^{m_x}$
  which satisfy $H_x H_z^T = 0$.
$n$, $m_x$, and $m_z$ are the number of qubits, the number of $X$ and $Z$ checks (i.e. stabilizer generators), respectively.
The code consists of $X$ and $Z$ logical operators represented by $C_x$ and $C_z$,
  $X$ and $Z$ stabilizers represented by $C_z^\perp$ and $C_x^\perp$,
  and nontrivial $X$ and $Z$ logical operators represented by $C_x - C_z^\perp$ and $C_z - C_x^\perp$.
The dimension $k = \dim C_x - \dim C_z^\perp$ is the number of logical qubits.
The distance is $d = \min(d_x, d_z)$ where
\begin{equation}
  d_x = \min_{c_x \in C_x - C_z^\perp} |c_x|,\qquad
  d_z = \min_{c_z \in C_z - C_x^\perp} |c_z|
\end{equation}
  are called the $X$ and $Z$ distances
  which are the minimal weights of the nontrivial $X$ and $Z$ logical operators.
The energy barrier is $\cE = \min(\cE_x, \cE_z)$
  where $\cE_x$ and $\cE_z$ are the minimum energies required to create a nontrivial $X$ and $Z$ logical operator
\begin{equation}
  \cE_x = \min_{\gamma_{0\to c_x}, c_x \in C_x - C_z^\perp} \epsilon_x(\gamma_{0\to c_x}),\qquad
  \cE_z = \min_{\gamma_{0\to c_z}, c_z \in C_z - C_x^\perp} \epsilon_z(\gamma_{0\to c_z}),
\end{equation}
  where $\epsilon_x(c_x) = |H_x c_x|$ is the number of violated $Z$-checks of the $X$ Pauli operator $c_x$
  and $\epsilon_z(c_z) = |H_z c_z|$ is the number of violated $X$-checks of the $Z$ Pauli operator $c_z$.

We say the quantum code is a low-density parity-check (LDPC) code if each check interacts with a bounded number of qubits and each qubit interacts with bounded number of checks, i.e. $H_x$ and $H_z$ have bounded number of nonzero entries in each column and row.

We remark that a quantum CSS code naturally corresponds to a chain complex.
In particular, $H_x$ and $H_z$ induce the chain complex
\begin{equation}
  X: \F2^{m_x} \xrightarrow{\delta_0 = H_z^T} \F2^n \xrightarrow{\delta_1 = H_x} \F2^{m_z}.
\end{equation}
In reverse, a chain complex also defines a quantum CSS code.
We denote the corresponding quantum code as $Q(X)$.
Later on, $X$ will mainly stands for the chain complex and not the $X$ in $X$-check.

\subsubsection{Geometric Properties: Geometrically Local and Bounded Density}

In addition to the linear algebraic structure from the parity-check matrices,
  the codes studied in this paper additionally have geometric structures.
In particular, the code is embedded in $\ZZ^D$\footnote{Instead of studying the embedding in $\ZZ^D$, one can study the embedding in $\RR^D$ which could be the more natural choice. We choose to work with $\ZZ^D$ for the ease of defining the notion of density. Morally, there is no difference between the two settings.}
  where each bit/qubit and check has a corresponding location in $\ZZ^D$
\begin{equation}
  I_{\code\to\euclid}: \Set(C) \to \ZZ^D \quad\text{or}\quad \Set(Q) \to \ZZ^D,
\end{equation}
  where $\Set(C)$ and $\Set(Q)$ are the sets that label the bits/qubits and the checks.
For classical codes $\Set(C) = [n] \sqcup [m]$
  and for quantum codes $\Set(Q) = [m_x] \sqcup [n] \sqcup [m_z]$,
  where $[n]$ is a set of size $n$ which labels the basis vectors of $\F2^n$
  and $\sqcup$ is the disjoint union.

We say the embedding is $a$-geometrically-local\footnote{We use the adverb ``geometrically'' to distinguish ``$a$-geometrically-local'' from ``$a'$-local'', which sometimes denotes the maximum number of qubits each check interacts with. Throughout this paper, local always means geometrically local.} if the Euclidean distance between each check and the qubit it interacts with is $\le a$.
That is, if $H_{ij}$ is nonzero, then $|I_{\code\to\euclid}(i)-I_{\code\to\euclid}(j)| \le a$ where $|\cdot|$ is the Euclidean distance.
We say the embedding has density $b$ if the number of qubits and checks located at each lattice point in $\ZZ^D$ is $\le b$.
The goal is to have constant parameters $a = \Theta(1)$, $b = \Theta(1)$.

\subsection{Good Quantum LDPC Code and Good Classical LTC from Balanced Product}
\label{sec:prelim-good-qldpc}

We now extract and review the key features of the good quantum LDPC codes constructed in \cite{panteleev2021asymptotically,lin2022good,dinur2022good}\footnote{Note that \cite{lin2022good} is not a full construction and is based on a conjecture on the existence of lossless expanders.}
  that allow us to transform them into geometrically local codes.
While each paper has its preferred way of describing their construction,
  here, we will put all of them under the umbrella of the balanced product \cite{breuckmann2021balanced} (also known as $G$-lifted product in \cite{panteleev2021asymptotically}),
  since it is the structure that allows us to put the code on a 2D square complex in a simple way.\footnote{What is actually needed is something weaker than balanced product. All we need is that the local view, i.e. the neighbors of an element, has a product structure. An example that has this structure but does not come from balanced product is the cubical complex \cite{jordan2000ramanujan}. Still we will stick with balanced product in this paper for simplicity.}
We will start by reviewing the balanced product construction,
  then discuss how the above good codes can be viewed as a balanced product code.

\subsubsection{Balanced Product of Codes}
Given two codes with invariant group actions $G$,
  the associated balanced product code is the tensor product of the two codes quotiented by the diagonal group action.
More precisely, let $H_A: \F2^{X_A(0)} \to \F2^{X_A(1)}$ and $H_B: \F2^{X_B(0)} \to \F2^{X_B(1)}$ be the parity-check matrices of the two codes.
The group $G$ act on sets $X_A(0), X_A(1), X_B(0), X_B(1)$
  such that $H_A$ and $H_B$ are invariant,
  i.e. $\<\wh a_1, H_A \wh a_0\> = \<\wh{g a_1}, H_A \wh{g a_0}\>$ where $\wh a_0$ ($\wh a_1$) is a basis vector labeled by an element $a_0 \in X_A(0)$ ($a_1 \in X_A(1)$) and $\<\cdot, \cdot\>$ is the inner product.
This implies an invariant group action of $G \times G$ on $H_A \otimes H_B: \F2^{X_A(0) \times X_B(0)} \xrightarrow{\delta_0} \F2^{X_A(1) \times X_B(0) \cup X_A(0) \times X_B(1)} \xrightarrow{\delta_1} \F2^{X_A(1) \times X_B(1)}$.

We then quotient $H_A \otimes H_B$ by the diagonal group action by identifying the coordinate $(a_i,b_j)$ and $(ga_i,gb_j)$ for $a_i \in X_A(i), b_j \in X_B(j), g \in G$.
More explicitly, we construct a new chain complex $H_A \otimes_G H_B: \F2^{X_A(0) \times_G X_B(0)} \xrightarrow{\delta_0'} \F2^{X_A(1) \times_G X_B(0) \cup X_A(0) \times_G X_B(1)} \xrightarrow{\delta_1'} \F2^{X_A(1) \times_G X_B(1)}$
  where $X_A(i) \times_G X_B(j) = (X_A(i) \times X_B(j)) / G$
  and each element is labeled by the equivalence class $\<[(a_i, b_j)]\> = \set{(g a_i, g b_j) : g \in G}$.
The new maps then inherit the entry values from the old map where
  $\<\wh{[(a_k, b_l)]}, \delta_h' \wh{[(a_i, b_j)]}\> = \<\wh{(a_k, b_l)}, \delta_h \wh{(a_i, b_j)}\>$.
  This map $\delta_h'$ is well-defined because $\delta_h$ is invariant under $G \times G$.\footnote{Note that the balanced product code construction is a special instance of the tensor product of modules where the underlying ring is the group ring $\F2[G]$. Such an operation is known as the balanced product for modules, which explains the origin of the name balanced product code.}

\subsubsection{Balanced Product of Graphs}
\label{sec:prelim-balanced-product-graph}
A related construction is the balanced product of graphs.
Given two graphs with invariant group actions $G$,
  the associated balanced product square complex is the Cartesian product of the two graphs quotiented by the diagonal group action.
More precisely, let $\cG_A = (V_A(0), V_A(1), E_A)$ and $\cG_B = (V_B(0), V_B(1), E_B)$ be the bipartite graphs where $V_A(0), V_A(1)$ are the vertex sets and $E_A$ is the edge set of $\cG_A$, and similarly for $\cG_B$.
The group $G$ acts on sets $V_A(0), V_A(1), V_B(0), V_B(1)$,
  such that if $(a_0, a_1) \in E_A$ then $(g a_0, g a_1) \in E_A$
    and if $(b_0, b_1) \in E_B$ then $(g b_0, g b_1) \in E_B$.
This implies an invariant group action $G \times G$ on $\cG_A \times \cG_B$ which is a square complex with
  vertices $V_A \times V_B$,
  edges $V_A \times E_B \cup E_A \times V_B$,
  and faces $E_A \times E_B$,
  where $V_A = V_A(0) \cup V_A(1)$ and $V_B = V_B(0) \cup V_B(1)$.

We then quotient $\cG_A \times \cG_B$ by the diagonal group action by identifying the vertices $(a_i,b_j)$ and $(ga_i,gb_j)$ for $a_i \in V_A(i), b_j \in V_B(j), g \in G$.
One can check that the resulting structure behaves nicely under this quotient and remains a square complex.

Readers may have noticed the similarities between the two constructions above.
In particular, each balanced product code has a natural embedding into an associated square complex.
Let $\cG_A$ be the graph where $V_A(0) = X_A(0), V_A(1) = X_A(1)$
  and $(a_0, a_1) \in E_A$ if the corresponding entry in $H_A$ is nonzero, $\<\wh a_1, H_A \wh a_0\> \ne 0$.
Define $\cG_B$ similarly.
Then the set of basis labels of $H_A \otimes_G H_B$, $\Set(H_A \otimes_G H_B) = \cup_{i=\set{0,1}, j=\set{0,1}} X_A(i) \times_G X_B(j)$, is precisely the vertex set of the square complex $\cG_A \times_G \cG_B$, $V_A \times_G V_B$.
In short, $\Set(X) = V(\cG)$ where $X \coloneqq H_A \otimes_G H_B$ and $\cG \coloneqq \cG_A \times_G \cG_B$.
This natural embedding is precisely what allows us to bypass the manifold and get a direct embedding into a 2D square complex (which is also a simplicial complex by cutting each square into two triangles).

\subsubsection{Good Codes as Balanced Product}

We now interpret the constructions in \cite{panteleev2021asymptotically,lin2022good,dinur2022good} as balanced product codes.
The construction in \cite{lin2022good} was described in terms of balanced product,
  so we will focus our attention on \cite{panteleev2021asymptotically,dinur2022good}
  which are based on Tanner codes of the left-right Cayley complex \cite{dinur2021locally}.
The way to view them in terms of balanced product has been discussed in \cite{panteleev2021asymptotically}.
Here, we only provide a short summary.

The code construction in \cite{panteleev2021asymptotically,dinur2022good} is to take a left-right Cayley complex $Cay^2(A, G, B)$, and then extend it using two classical codes $C_A, C_B$.
The (4-fold) left-right Cayley complex $Cay^2(A, G, B)$ is specified by a finite group $G$ and two sets of generators $A$ and $B$.
Structurally speaking, the left-right Cayley complex is a square complex, which is similar to a simplicial complex, but instead consists of vertices, edges, and squares:
\begin{itemize}
  \item The vertices are $V = V_{00} \cup V_{10} \cup V_{01} \cup V_{11}$ where $V_{00} \cong V_{10} \cong V_{01} \cong V_{11} \cong G$.
  \item The edges consists of vertical and horizontal edges $E = E^- \cup E^| = (E_{*0} \cup E_{*1}) \cup (E_{0*} \cup E_{1*})$ where
    \begin{align*}
      E_{*0} &= \{(g, ag) : g \in G, a \in A\} \subset V_{00} \times V_{10}\;, \\
      E_{*1} &= \{(gb, agb) : gb \in G, a \in A\} \subset V_{01} \times V_{11}\;, \\
      E_{0*} &= \{(g, gb) : g \in G, b \in B\} \subset V_{00} \times V_{01}\;, \\
      E_{1*} &= \{(ag, agb) : ag \in G, b \in B\} \subset V_{10} \times V_{11}\;.
    \end{align*}
  \item The faces are $F = \{(g, ag, gb, agb): g \in G, a \in A, b \in B\} \subset V_{00} \times V_{10} \times V_{01} \times V_{11}$\;.
\end{itemize}

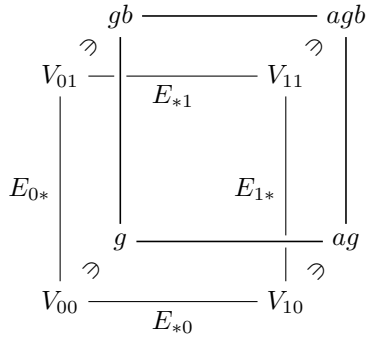
\begin{figure}[H]
  \centering
  \begin{tikzpicture}
    \node (V00) at (0,0) {$V_{00}$};
    \node (V10) at (3,0) {$V_{10}$};
    \node (V01) at (0,3) {$V_{01}$};
    \node (V11) at (3,3) {$V_{11}$};

    \draw (V00) to node[auto, swap] {$E_{*0}$} (V10);
    \draw (V00) to node[auto] {$E_{0*}$} (V01);
    \draw (V01) to node[auto, swap] {$E_{*1}$} (V11);
    \draw (V10) to node[auto] {$E_{1*}$} (V11);

    \draw (0,0)+(0.8,0.8) node (g) {$g$};
    \draw (3,0)+(0.8,0.8) node (ag) {$ag$};
    \draw (0,3)+(0.8,0.8) node (gb) {$gb$};
    \draw (3,3)+(0.8,0.8) node (agb) {$agb$};

    \path (g) -- node [sloped] {$\ni$} (V00);
    \path (ag) -- node [sloped] {$\ni$} (V10);
    \path (gb) -- node [sloped] {$\ni$} (V01);
    \path (agb) -- node [sloped] {$\ni$} (V11);

    \draw [white, double=black, line width=2pt] (g) to (ag);
    \draw [white, double=black, line width=2pt] (g) to (gb);
    \draw [white, double=black, line width=2pt] (gb) to (agb);
    \draw [white, double=black, line width=2pt] (ag) to (agb);
  \end{tikzpicture}
  \caption{4-fold left-right Cayley complex.}
  \label{fig:left-right-cayley-complex}
\end{figure}

To interpret the code as a balanced product, the key observation is that $Cay^2(A, G, B) \cong Cay(A, G) \times_G Cay(G, B)$
  where $Cay(A, G)$ ($Cay(G, B)$) is the Cayley graph with generators $A$ ($B$) acting through left (right) multiplication, i.e. the vertices are labeled by $g \in G$ where $g$ and $ag$ ($gb$) are adjacent.
We set $Cay(A, G)$ to have group acting from the right with $g \cdot x = xg^{-1}$ for $x \in V(Cay(A, G))$,
  and $Cay(G, B)$ to have group acting from the left with $g \cdot y = gy$ for $y \in V(Cay(G, B))$.
The reason for this peculiar group action is that
  this choice gives the bijection between the vertices of $Cay(A, G) \times_G Cay(G, B)$ and the vertices $Cay^2(A, G, B)$ through
  $[(x, y)] \mapsto xy$
  where $[(x, y)]$ is the equivalence class $\set{(g \cdot x,g \cdot y): g \in G}$.
One can check that this is well-defined
  because $[(x, y)] = [(g \cdot x, g \cdot y)] \mapsto x g^{-1} g y = xy$.
It is also straightforward to check that this map between the vertices also maps correctly between the edges and the faces.

From this observation, we claim that the qLDPC code in \cite{panteleev2021asymptotically,dinur2022good} is a balanced product code of the form $\cT(Cay(A, G), C_A) \otimes_G \cT(Cay(G, B), C_B)$,
  where $\cT(Cay(A, G), C_A)$ ($\cT(Cay(G, B), C_B)$) is the Tanner code with base graph $Cay(A, G)$ ($Cay(G, B)$) extended using the local code $C_A$ ($C_B$).
The main idea is that the step of Tannerization is independent of the group structure.
  So the property of being a balanced product code
  boils down to the graph $Cay^2(A, G, B)$ having a balanced product structure $Cay(A, G) \times_G Cay(G, B)$,
  which we have discussed above.
We will leave other details to \cite{panteleev2021asymptotically}.

\subsubsection{Good Codes}
From \cite{panteleev2021asymptotically,dinur2022good}
  it is known that there exist good codes from balanced products.
Let $X$ be the chain complex constructed in the papers above
\begin{equation}
  X: \F2^{X(0)} \to \F2^{X(1)} \to \F2^{X(2)}.
\end{equation}
It has been shown that the corresponding quantum codes have linear dimension, linear distance, linear energy barrier.\footnote{The property of having linear energy barrier is not stated explicitly in the papers. But as we will see in \Cref{sec:prelim-chain-complex-expansion}, energy barrier can be inferred from the small-set boundary expansion studied in \cite{dinur2022good}.}
\begin{theorem} [\cite{panteleev2021asymptotically,dinur2022good}]
  \label{thm:qLDPC}
  $Q(X)$ provide families of quantum LDPC codes with linear dimension, linear distance, and linear energy barrier, i.e. $k = \Theta(n)$, $d = \Theta(n)$, and $\cE = \Theta(n)$.
\end{theorem}

Furthermore, as discussed above, the elements of the code can be identified with the vertices of the square complex.
This maps is denoted as
\begin{equation} \label{eq:code-to-square}
  I_{\code\to\sq}: \Set(X) \to V(\cG)
\end{equation}
where $\Set(X) = X(0) \cup X(1) \cup X(2)$ is the set of labels for the basis vectors in the three vector spaces
  and $\cG$ is the square complex.
See \Cref{fig:graph-and-code} for an example of the square complex and the balanced product code.

\begin{figure}[H]
  \centering
  \includegraphics[width=0.9\textwidth]{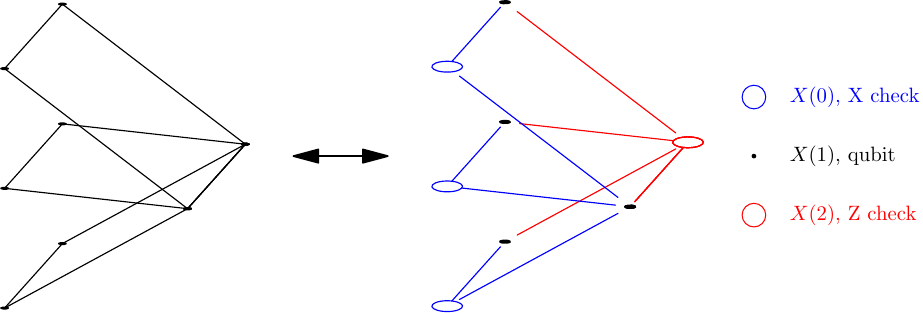}
  \caption{An example of the square complex $\cG$ and the balanced product code $X$.}
  \label{fig:graph-and-code}
\end{figure}

Therefore, to obtain an embedding of the code in $\ZZ^D$,
  it suffices to obtain an embedding of the square complex in $\ZZ^D$.
Unfortunately, $\cG$ is an expander and expanders cannot be embedded locally in $\ZZ^D$.
Nevertheless, the subdivided square complex $\cG_L$ introduced later has a local embedding
  which is how we embed the subdivided code $X_L$ in $\ZZ^D$.
This will be discussed later in \Cref{sec:prelim-random-embedding}.

\subsection{Expansion Properties of Chain Complexes}
\label{sec:prelim-chain-complex-expansion}

In the previous subsection, we claimed that the quantum code $Q(X)$ has linear distance and linear energy barrier.
To discuss these concepts at once, it is helpful to introduce the notion of small-set (co)boundary expansion initially discussed in \cite{hopkins2022explicit} and further enhanced in \cite{dinur2022good}. These expansion properties are natural generalizations of the (co)boundary expansion \cite{linial2006homological}, which will also be used later in the paper.

\begin{definition} [Small-Set (Co)Boundary Expansion]
  We say that $X: \F2^{X(0)} \xrightarrow{\delta_0} \F2^{X(1)} \xrightarrow{\delta_1} \F2^{X(2)}$ is a $(\alpha,\beta,\gamma)$-small-set boundary expander if
  \begin{equation*}
    \forall c_1 \in \F2^{X(1)}, \abs{c_1} \le \alpha |X(1)|:
    \exists c_2 \in \F2^{X(2)}, \beta \abs{c_1 + \partial_2 c_2} \le \abs{\partial_1 c_1}, \gamma \abs{c_2} \le \abs{c_1}.
  \end{equation*}

  Similarly, $X$ is a $(\alpha,\beta,\gamma)$-small-set coboundary expander if
  \begin{equation*}
    \forall c_1 \in \F2^{X(1)}, \abs{c_1} \le \alpha |X(1)|:
    \exists c_0 \in \F2^{X(0)}, \beta \abs{c_1 + \delta_0 c_0} \le \abs{\delta_1 c_1}, \gamma \abs{c_0} \le \abs{c_1}.\footnote{In the standard literature on HDX, it is customary to normalize the vector by the size of the set, i.e. $\norm{c_i} = \abs{c_i} / |X(i)|$.
    The advantage of this convention is that it makes $\beta = 1$ a natural unit. See \cite[App B]{dikstein2023coboundary}.
    While we would like to follow this convention, doing so would introduce additional scaling factors that complicate the equations later on.
    Therefore, we will adhere to this unconventional normalization.} \footnote{We have slightly modified the definition by replacing $\gamma$ with $1/\gamma$ compared to the earlier version of \cite{dinur2022good}. The motivation behind this modification is that in the usual convention, a larger parameter is harder to achieve, whereas in the original definition larger $\gamma$ is easier to achieve. Our new definition aligns with the convention that larger parameters are more challenging to obtain.}
  \end{equation*}
\end{definition}

Roughly speaking, $\alpha$ and $\beta$ correspond to distance and energy barrier, respectively.
More precisely, if the chain complex $X$ has $(\alpha,\beta,\gamma)$-small-set boundary and coboundary expansion, then the corresponding quantum code $Q(X)$ has distance $> \alpha |X(1)| = \alpha n$ and energy barrier $\ge \alpha \beta |X(1)| = \alpha \beta n$.\footnote{Another place where $\beta$ shows up is in the proof of the NLTS conjecture \cite[Property 1]{anshu2022nlts} with $c_1 = \frac{\max(m_x, m_z)}{\beta n}$ and $c_2 = \alpha$. The relation is that, when $|y|_{C_x^\perp} < \alpha n$, small-set (co)boundary expansion applies and one gets $\beta |y|_{C_x^\perp} \le \text{number of violated checks}$. Therefore, when the number of violated checks is small, $\le \delta m_x$, the vector has a small weight proportional to the number of violated checks, $|y|_{C_x^\perp} \le c_1 \delta n$, or has a large weight, $|y|_{C_x^\perp} \ge c_2 n$. This property is also recently used in \cite{gu2024single}.}\footnote{Readers might wonder if $\gamma$ is related to a code property. Indeed, it is related to the soundness of the classical code with parity check matrix $\delta_0$. We will not go into the details, but one can show that if the chain complex has coboundary expansion, then the classical code has soundness $\ge \min(\alpha, \frac{|X(0)|}{|X(1)|}\gamma) = \min(\alpha, \frac{n}{m}\gamma)$ where $m$ is the number of checks.}

Here is a sketch of the bound on the energy barrier.
\begin{proof} [Proof of the bound on energy barrier from small-set coboundary expansion]
  Recall the energy barrier is the minimal number of violated checks needed to go from $0$ to a nonzero codeword.
  Let $\norm{c_1}$ be the minimal weight among the vectors in the same homology as $c_1$, $\norm{c_1} = \min_{c_0 \in \F2^{X(0)}} |c_1 + \delta_0 c_0|$.
  Since a nonzero codeword has weight $\norm{c_1} > \alpha|X(1)|$,
    along the path from $0$ to a nonzero codeword, there has to be a vector of weight $\norm{c_1'} = \alpha|X(1)|$.
  Let $c_1''$ to be the vector of minimal weight in the same homology as $c_1'$.
    Since they are in the same homology, they have the same number of violated checks,
    $|\delta_1 c_1''| = |\delta_1 c_1'|$.
  And because $|c_1''| = \norm{c_1'} = \alpha|X(1)|$,
    small-set coboundary expansion applies
    and we have $\beta |c_1''| \le |\delta_1 c_1''|$.
  Therefore, the energy barrier is at least $\alpha \beta |X(1)|$,
    since we have identified a vector $c_1'$ on the path with a large number of violated checks $|\delta_1 c_1'| = |\delta_1 c_1''| \ge \beta |c_1''| = \alpha \beta |X(1)|$.
\end{proof}

Therefore, \Cref{thm:qLDPC} about distance and energy barrier for $Q(X)$ can be obtained as a corollary of the following theorem, the way it was done in \cite{dinur2022good}.
\begin{theorem}
  $X$ provides a family of bounded degree 3-term chain complexes with $(\alpha,\beta,\gamma)$-small-set boundary expansion and coboundary expansion for certain constants $\alpha,\beta,\gamma$.
\end{theorem}
Here, bounded degree means LDPC, i.e. $\delta_0$ and $\delta_1$ have bounded number of nonzero entries in each column and row.

In the special case when $\alpha = 1$, $|c_1| \le \alpha |X(1)| = |X(1)|$ imposes no constraints on $c_1$ and this is known as the boundary expansion.
\begin{definition} [(Co)Boundary Expansion] \label{def:standard-coboundary-expansion}
  We say that $X: \F2^{X(0)} \xrightarrow{\delta_0} \F2^{X(1)} \xrightarrow{\delta_1} \F2^{X(2)}$ is a $\beta$-boundary expander if
  \begin{equation*}
    \forall c_1 \in \F2^{X(1)}:
    \exists c_2 \in \F2^{X(2)}, \beta \abs{c_1 + \partial_2 c_2} \le \abs{\partial_1 c_1}.
  \end{equation*}

  Similarly, $X$ is a $\beta$-coboundary expander if
  \begin{equation*}
    \forall c_1 \in \F2^{X(1)}:
    \exists c_0 \in \F2^{X(0)}, \beta \abs{c_1 + \delta_0 c_0} \le \abs{\delta_1 c_1}.
  \end{equation*}
\end{definition}

Note that when $X$ is a boundary expander $X$ has to be exact.
Suppose $c_1 \in Z_1 = \Ker \partial_1$, i.e. $\partial_1 c_1 = 0$, the inequality $0 = \abs{\partial_1 c_1} \ge \beta \abs{c_1 + \partial_2 c_2}$ implies $c_1 + \partial_2 c_2 = 0$, i.e. $c_1 = \partial_2 c_2 \in B_1$.
This means $B_1 = Z_1$, i.e. exact.

To fit our proof method, we will later extend the definition of coboundary expansion in \Cref{def:coboundary-expansion}.

\subsection{Subdivision of a Square Complex and its Random Embedding in $\ZZ^D$}
\label{sec:prelim-random-embedding}

The overall strategy to construct our geometrically local codes is by composing two maps, one from the code to the square complex, $I_{\code\to\sq}$, the other from the square complex to the Euclidean space $\ZZ^D$, $I_{\sq\to\euclid}$.
We discussed aspects of $I_{\code\to\sq}$ in \Cref{sec:prelim-good-qldpc}.
This section focuses on $I_{\sq\to\euclid}$.
The strategy of the embedding follows from \cite{gromov2012generalizations,portnoy2023local}.

Because the square complex is an expander graph which does not have a geometrically local embedding,
  we have to subdivide it which then has a local embedding.

\begin{definition} [Subdivide a Square Complex]
  Given a square complex $\cG = (V, E, F)$,
  the $L$-subdivision $\cG_L = (V_L, E_L, F_L)$ is the square complex where
  each face is subdivided into $L \times L$ squares
  and the edges and vertices are included accordingly.
\end{definition}
Two examples are illustrated in \Cref{fig:subdivide-square}.

\begin{figure}[H]
  \centering
  \includegraphics[width=0.65\textwidth]{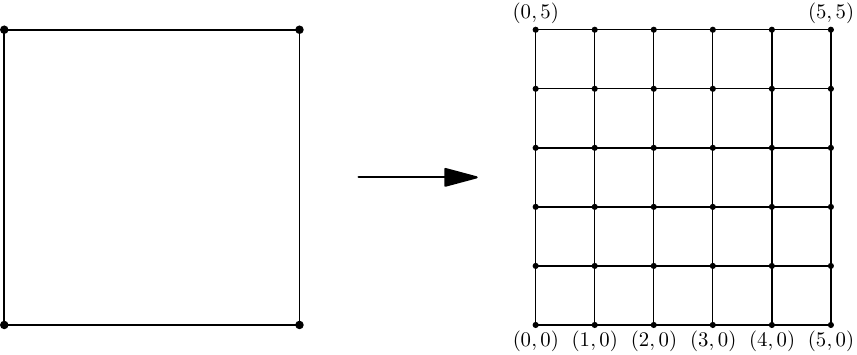} \\
  \vspace{3em}
  \includegraphics[width=0.8\textwidth]{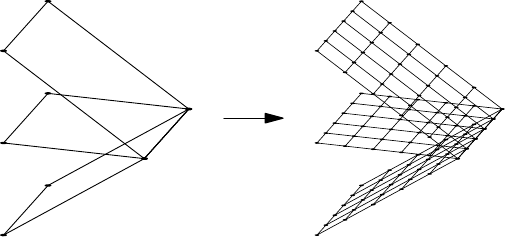}
  \caption{The subdivision of square complexes.}
  \label{fig:subdivide-square}
\end{figure}

For the ease of the subdivided code discussed later in \Cref{sec:construct-subdivided-code},
  we assign a coordinate to each vertex in $\cG_L$.
Each vertex has a coordinate $(i,j)$ for some $0 \le i, j \le L$,
  depending on its location in the square.
Recall that because $\cG$ is the balanced product of two bipartite graphs,
  the vertices of $\cG$ can be split into 4 groups,
  $V(\cG) = V_{00} \cup V_{10} \cup V_{01} \cup V_{11}$.
We assign the coordinate of the vertices in $V_{00}, V_{10}, V_{01}, V_{11}$ as $(0,0), (L,0), (0,L), (L,L)$, respectively.
Other vertices in $\cG_L$ belong to some face $(v_{00}, v_{10}, v_{01}, v_{11}) \in F$ and have a corresponding coordinate $(i,j)$.

It is known that for large enough $L$, the subdivided complex $\cG_L$ has a local embedding with bounded density.
\begin{theorem} [Modified {\cite[Theorem 7]{portnoy2023local}}]
  \label{thm:embedding-old}
  Given a bounded degree square complex $\cG$,
  the $L$-subdivision $\cG_L$ has an embedding $f: V(\cG_L) \to \ZZ^D$ with constant $a, b = \Theta(1)$ such that
  \begin{enumerate}
    \item Geometrically local: For all adjacent vertices on the complex $(v_0, v_1) \in E(\cG_L)$,
      the distance of corresponding points in $\ZZ^D$ is bounded $|f(v_0) - f(v_1)| \le a$.
    \item Bounded density: The number of vertices at each point is bounded,
      i.e. for all $x \in \ZZ^D$, $|f^{-1}(x)| \le b$.
  \end{enumerate}
  for $L = |V(\cG)|^{\frac{1}{D-2}} \log^{D+1} |V(\cG)|$.
\end{theorem}

We summarize the construction and the proof of the theorem.
The embedding is obtained by mapping the vertices in $V(\cG)$ randomly to the sphere of dimension $D-1$ with radius $L' = V^{\frac{1}{D-2}}$ where $V = |V(\cG)|$ is the number of vertices.
We then interpolate the points to obtain an embedding for the subdivided square complex with subdivision $L'$.
The image is currently in $\RR^D$, so we nudge the vertices slightly so that the image becomes $\ZZ^D$.
(We actually need an extra perturbation of the points which will be discussed shortly. This is where the factor $\log^{D+1} V$ appears in $L$.)

The construction automatically satisfies locality,
  so it suffices to check that the embedding has bounded density.
The idea is to take the union bound over the bad events
  where some unit ball contains too many vertices.
The expected number of vertices in a ball is roughly the volume of the square complex with thickness $1$ divided by the volume of the ball of radius $L'$.
  The volume of the square complex is $\Theta({L'}^2 V)$
  and the volume of the ball of radius $L'$ is $\Theta({L'}^D)$.
Therefore, with the choice $L' = V^{\frac{1}{D-2}}$, the expected number of vertices in a ball is constant.
Using Chernoff bound, the probability that the ball has $\Omega(\log V)$ vertices is then small enough to apply the union bound.
However, $\Omega(\log V)$ is not bounded density, so we have to do an extra perturbation and this is why the polylog factor appears in $L$.
The full detail can be found in \cite[App 1]{portnoy2023local}.

We will improve this embedding result in an upcoming work which is optimal, removing the polylog factor, and providing an explicit embedding without randomness.
\begin{theorem} [Upcoming]
  \label{thm:embedding}
  Given a bounded degree square complex $\cG$,
  the $L$-subdivision $\cG_L$ has an embedding $f: V(\cG_L) \to \ZZ^D$ with constant $a, b = \Theta(1)$ such that
  \begin{enumerate}
    \item Geometrically local: For all adjacent vertices on the complex $(v_0, v_1) \in E(\cG_L)$,
      the distance of corresponding points in $\ZZ^D$ is bounded $|f(v_0) - f(v_1)| \le a$.
    \item Bounded density: The number of vertices at each point is bounded,
      i.e. for all $x \in \ZZ^D$, $|f^{-1}(x)| \le b$.
  \end{enumerate}
  for $L = \Theta(|V(\cG)|^{\frac{1}{D-2}})$
    where the constant depends only on the degree of the complex and the dimension $D$.
\end{theorem}
We denote this embedding $f$ as $I_{\sq_L \to \euclid}$ which will be used to embed the subdivided code into $\ZZ^D$.

\subsection{Surface Codes}
\label{sec:prelim-surface-codes}

The surface code \cite{kitaev2003fault,eczoo_surface}, also known as the toric code, is a CSS code which is geometrically local and has a planar layout, making it highly suitable for practical quantum computing.
In our context, the surface code is the simplest example of the subdivided code that we will define later.
Furthermore, many properties of the subdivided code boils down to the properties of the surface code (or more precisely, the generalized surface code), so we give a brief review of the surface code.

Given a subdivided square of size $L \times L$,
  we can arrange the qubits and the checks based on whether the coordinate $(i,j)$ is even or odd as illustrated in \Cref{fig:surface-code-intro}.
It is known that the surface code has parameters $[[n=(\frac{L}{2}+1)^2+(\frac{L}{2})^2, k=1, d=\frac{L}{2}+1]]$.\footnote{Note that the layout of the qubits and the checks for the subdivided code will look slightly different from the layout for the surface code. In particular, $L$ will be odd for the subdivided code, in contrast to the case here where $L$ is even.}

However, knowing the distance of the surface code is not enough for the proof we will present later.
What we will need is the coboundary expansion properties similar to those defined in \Cref{sec:prelim-chain-complex-expansion}.
Furthermore, we will have to work with codes that are more general than the surface codes which we call the generalized surface codes.
Roughly speaking, these are codes with 2D structures, just like the surface codes, but with extra branching points, similar to the square complexes.
The generalized surface codes and their expansion properties will be discussed in \Cref{sec:prop-surface-code}.
\begin{figure}[H]
  \centering
  \includegraphics[width=0.9\textwidth]{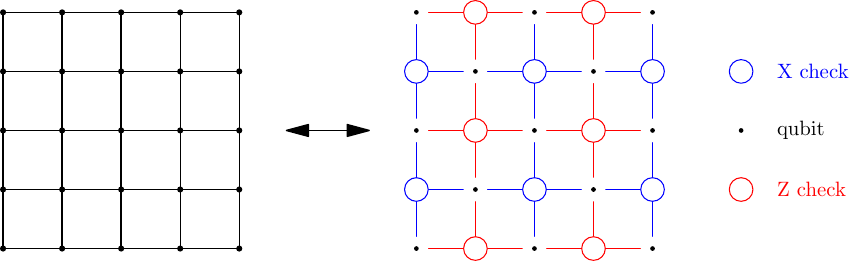}
  \caption{Surface code of size $L=4$.}
  \label{fig:surface-code-intro}
\end{figure}

\section{Construct Geometrically Local Codes through Subdivision}
\label{sec:construct-subdivided-code}

Equipped with the knowledge that the balanced product code $X$ has a square complex structure $\cG$ as discussed in \Cref{sec:prelim-good-qldpc},
  we now use this structure to construct the subdivided code $X_L$
  which has a natural embedding $I_{\code_L\to\sq_L}: \Set(X_L) \to V(\cG_L)$.
After composing with the map $I_{\sq_L\to\euclid}$ in \Cref{sec:prelim-random-embedding},
  this gives the desired local embedding of the subdivided code $I_{\code_L\to\euclid}: \Set(X_L) \to \ZZ^D$.

Take an odd number $L$.
We split the vertices of $V(\cG_L)$ into three groups $X_L(0), X_L(1), X_L(2)$ which will be the basis for $X_L$.
Recall that each vertex in $V(\cG_L)$ has a corresponding coordinate $(i,j)$.
We now define
\begin{itemize}
  \item $X_L(0)$ to be the set of vertices with $i,j$ both being even.
  \item $X_L(1)$ to be the set of vertices with $i,j$ one being even, one being odd.
  \item $X_L(2)$ to be the set of vertices with $i,j$ both being odd.
\end{itemize}

The linear maps are defined by the adjacency matrices,
  where $X_L(0) \xrightarrow{\delta_0} X_L(1)$ is the adjacency matrix between $X_L(0)$ and $X_L(1)$
  and $X_L(1) \xrightarrow{\delta_1} X_L(2)$ is the adjacency matrix between $X_L(1)$ and $X_L(2)$.
One can check that this defines a chain complex.
Two examples are illustrated in \Cref{fig:subdivide-1,fig:subdivide-2}.

\begin{remark}
  In the main text, we describe the construction as a subdivision process.
  An alternative way to understand the construction is through welding \cite{michnicki20143d}.
  Welding is a process that glues multiple surface codes into a larger code.
  Our construction is essentially multiple copies of the surface codes on the faces of the balanced product graph
    welded along the edges.
  The extra ingredient compared to \cite{michnicki20143d}
    is that our welding is built on top of the good qLDPC codes
    which have better code properties.
\end{remark}

\begin{figure}[H]
  \centering
  \includegraphics[width=0.65\textwidth]{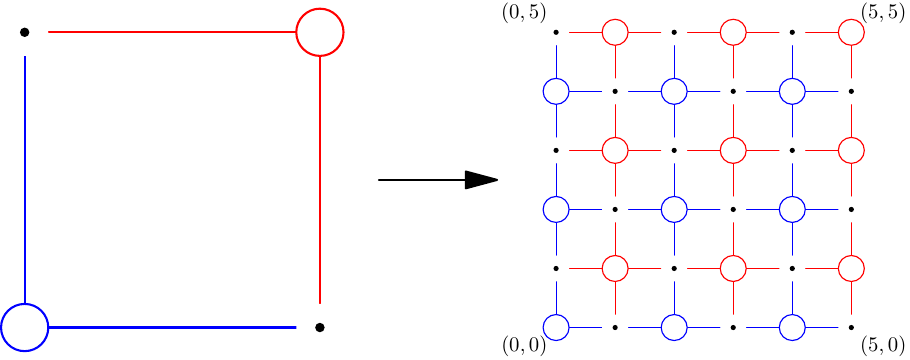}
  \caption{The subdivision of a balanced product code.}
  \label{fig:subdivide-1}
\end{figure}

\begin{figure}[H]
  \centering
  \includegraphics[width=0.65\textwidth]{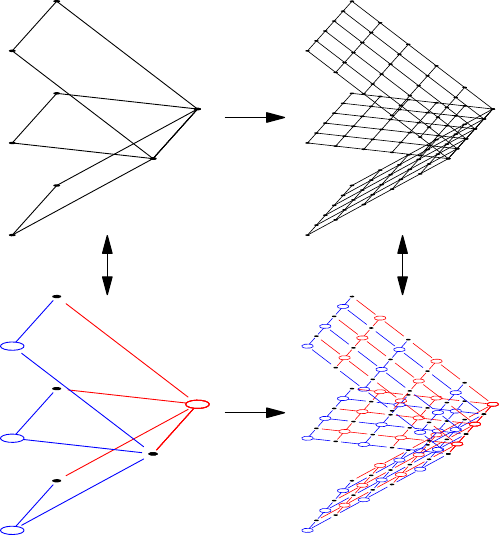}
  \caption{The subdivision of a balanced product code together with the subdivision of the associated square complex.}
  \label{fig:subdivide-2}
\end{figure}

\subsection{Full Construction of the Geometrically Local Quantum Codes}
\label{sec:construct-full-code}

We now have all the ingredients needed to describe the full construction of the geometrically local codes.
Recall from \Cref{sec:prelim-codes} that to describe a geometrically local code in $D$ dimension, we need to specify the parity-check matrices of the code together with an embedding of the qubits and the checks in $\ZZ^D$.

We first describe the parity-check matrices.
Let $X$ be the chain complex from \cite{panteleev2021asymptotically,dinur2022good} that has small-set boundary and coboundary expansion.
Let $X_L$ be the chain complex after $L$-subdivision defined in the previous subsection.
We will set the value of $L$ in a moment.
The quantum code $Q(X_L)$ is defined to have parity-check matrices $H_z = \delta_0^T$ and $H_x = \delta_1$.

For the embedding, with the appropriate choice of $L = \Theta(|V(\cG)|^{\frac{1}{D-2}})$ from \Cref{thm:embedding}, we have the embedding map $I_{\sq_L\to\euclid}: V(\cG_L) \to \ZZ^D$ that is geometrically local and has bounded density.
Here, $\cG$ is the square complex associated to $X$
  and $\cG_L$ is the subdivided square complex.
By composing the isomorphism $I_{\code_L\to\sq_L}: \Set(X_L) \xrightarrow{\sim} V(\cG_L)$ and $I_{\sq_L\to\euclid}$
  we have the embedding
\begin{equation}
  I_{\code_L\to\euclid} = I_{\sq_L\to\euclid} \,\circ\, I_{\code_L\to\sq_L}: \Set(X_L) \to \ZZ^D.
\end{equation}
The embedding $I_{\code_L\to\euclid}$ induces the embedding of the codes simply because $\Set(Q(X_L)) = \Set(X_L)$.
This completes the full construction of the geometrically local codes.

We discuss some technical aspect of the code.
In our discussion, we will represent the max degree of the bipartite graphs used for the balanced product as $\Delta_{\max}$
  and the min degree as $\Delta_{\min}$.
(These bipartite graphs $G_1, G_2$ are discussed in \Cref{sec:prelim-balanced-product-graph}.)
For the sake of making the parameters look simpler, we will use the fact that $\Delta_{\min} \ge 2$, which can be chosen for all the good qLDPC codes above.

\subsection{Subcomplexes of the Subdivided Code}
\label{sec:subcomplex}
The analysis of the subdivided code generally contains two parts.
One is about the original code which has been worked out in earlier works.
The other is about the subcomplexes that will be discussed in this subsection.

Consider the connected components of $\cG_L$ when removing all edges of the form $\set{(i,L-1), (i,L)}$ and $\set{(L-1,j), (L,j)}$.
We observe that each component contains exactly one of the corner vertices.
Additionally, these corner vertices correspond to the vertices in $\cG$.
This induces a bijection between the vertices in $\cG$ and the connected components in $\cG_L$.

To be formal, we split the vertices $V(\cG_L)$ into 3 groups $S, T, U$ where
\begin{itemize}
  \item $S$ contains vertices with coordinates $0 \le i, j \le L-1$.
  \item $T$ contains vertices with coordinates $i = L, 0 \le j \le L-1$ or $0 \le i \le L-1, j = L$.
  \item $U$ contains vertices with coordinates $i = j = L$.
\end{itemize}
We see that $S, T, U$ are not fully connected
  and let $\cS, \cT, \cU$ be the set of connected components of $S, T, U$ respectively.
Since each connected components contains exactly one of the corner vertices, we have the bijections $X(0) \cong \cS$, $X(1) \cong \cT$, $X(2) \cong \cU$.
See \Cref{fig:proof-regions}.

\begin{figure}[H]
  \centering
  \includegraphics[width=0.99\linewidth]{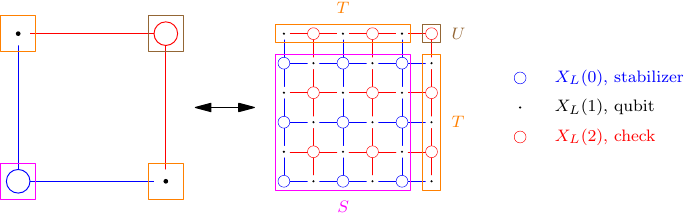}
  \vspace{1em}
  \caption{Regions $S$, $T$, $U$ and the bijections $X(0) \cong \cS$, $X(1) \cong \cT$, $X(2) \cong \cU$.}
  \label{fig:proof-regions}
\end{figure}

Notice that each component forms a chain complex.
Let $S_i \in \cS$ be one of the components.
Then the restriction $X_L|_{S_i}: \F2^{S_i \cap X_L(0)} \to \F2^{S_i \cap X_L(1)} \to \F2^{S_i \cap X_L(2)}$ is a chain complex. (See \Cref{fig:proof-regions}.)
One can similarly study the linear map (technically also a chain complex) $X_L|_{T_i}: \F2^{T_i \cap X_L(1)} \to \F2^{T_i \cap X_L(2)}$ for $T_i \in \cT$.
(Note that $T_i \cap X_L(0) = 0$ and $X_L|_{U_i}$ is just a vertex.)
These chain complexes (subcomplexes) will be essential for the analysis of the subdivided code.
Since $X_L|_{S_i}$ and $X_L|_{T_i}$ share similarities to the surface code and the repetition code,
  we call them the generalized surface code and the generalized repetition code
  which are depicted in \Cref{fig:generalized-code}.

We discuss a few features of these subcomplexes.
More of their properties will be studied in \Cref{sec:prop-surface-code}.
Notice that the generalized surface code has a product structure.
It is the tensor product of two generalized repetition codes.
This property comes from the fact that the balanced product code is locally a product.
The generalized repetition code is like a tree
  with a unique branching point.
Notice that the number of branches
  corresponds to the degree of some vertex in the bipartite graphs $G_1, G_2$.
As reminder, this number is between $\Delta_{\min} = 2$ and $\Delta_{\max} = \Theta(1)$.

\begin{figure}[H]
  \centering
  \includegraphics[width=0.8\linewidth]{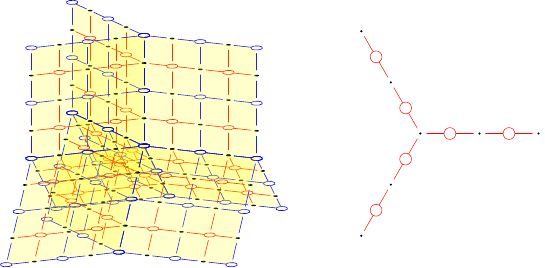}
  \vspace{1em}
  \caption{Generalized surface code and generalized repetition code.}
  \label{fig:generalized-code}
\end{figure}

We now define the chain maps $\cF_i: \F2^{X(i)} \to \F2^{X_L(i)}$ for $i=0,1,2$ that will be used repeatedly in the analysis later.
\begin{itemize}
  \item Given $\wt c_0 \in \F2^{X(0)}$, we define $\cF_0(\wt c_0) \in \F2^{X_L(0)}$
    by repeating the value $\wt c_0(v)$ at the corresponding component $S_i$ of $v$ for each $v \in X(0)$.
  \item Given $\wt c_1 \subset \F2^{X(1)}$, we define $\cF_1(\wt c_1) \in \F2^{X_L(1)}$
    by repeating the value $\wt c_1(v)$ at the corresponding component $T_i$ of $v$ for each $v \in X(1)$
    and set other values in $S$ to be $0$.
  \item Given $\wt c_2 \subset \F2^{X(2)}$, we define $\cF_2(\wt c_2) \in \F2^{X_L(2)}$
    by setting the value $\wt c_2(v)$ at the corresponding vertex $U_i$ of $v$ for each $v \in X(2)$
    and set other values in $S$ and $T$ to be $0$.
\end{itemize}
Notice that these maps commutes with $\delta$.
In particular,
\begin{equation}\label{eq:cF}
  \begin{tikzcd}
    \F2^{X(0)} \arrow[r, "\delta_0"] \arrow[d, "\cF_0"] &
    \F2^{X(1)} \arrow[r, "\delta_1"] \arrow[d, "\cF_1"] &
    \F2^{X(2)} \arrow[d, "\cF_2"] \\
    \F2^{X_L(0)} \arrow[r, "\delta_0"] & \F2^{X_L(1)} \arrow[r, "\delta_1"] & \F2^{X_L(2)}
  \end{tikzcd}
\end{equation}
Many of the later proofs will be moving back and forth between the two chain complexes through $\cF$.
Notice that $\cF_i$ are injective because by construction they are repeating values.

Applying the discussions above, we can show the following bounds on $X_L(i)$.
\begin{claim} \label{claim:X-bound}
  \begin{equation}
    L^2 |X(i)| \le |X_L(i)| \le \frac{\Delta_{\max}^2}{4} L^2 |X(i)|
  \end{equation}
  for all $i=0,1,2$.
  In particular, $|X_L(i)| = \Theta(L^2 |X(i)|)$.
\end{claim}
\begin{proof}
  We show the case for $i=0$.
  Other cases follow similarly by noticing the symmetry that
  $\cG$ is a 4-partite graph
    and $X(0), X(1), X(2)$ contain $1, 2, 1$ parts, respectively.

  To proof utilize the isomorphism $X(0) \cong \cS$.
  Notice that each connected component $S_i \in \cS$ contains $(\Delta_1 \frac{L-1}{2} + 1)(\Delta_2 \frac{L-1}{2} + 1)$ vertices in $X_L(0)$
    where $\Delta_1$ and $\Delta_2$ are the degrees of branchings.
    (One way to see this is that $S_i$ has a product structure and the each piece contains $\Delta_j \frac{L-1}{2} + 1$ vertices.)
  Since $2 \le \Delta_1, \Delta_2 \le \Delta_{\max}$,
    the number of vertices in each component is between $L^2$ and $(\Delta_{\max} \frac{L-1}{2} + 1)^2 \le \frac{\Delta_{\max}^2}{4} L^2$.
  Hence, we have the result.
\end{proof}

It is helpful to express the size of various sets in terms of $L$.
Let $V = |V(\cG)| = |X(0)| + |X(1)| + |X(2)|$.
Recall $L = \Theta(V^{\frac{1}{D-2}})$, i.e. $V = \Theta(L^{D-2})$.
It is known that $|X(0)|, |X(1)|, |X(2)| = \Theta(V) = \Theta(L^{D-2})$.
Combining with the claim above, we have
\begin{align}
  |X(0)|, |X(1)|, |X(2)| &= \Theta(L^{D-2}), \label{eq:X-size} \\
  |X_L(0)|, |X_L(1)|, |X_L(2)| &= \Theta(L^D). \label{eq:XL-size}
\end{align}

\section{Properties of the Subdivided Code}
\label{sec:prop-subdivided-code}

Having constructed the codes in the previous section, we now study their properties and show \Cref{thm:quantum-LDPC}.
We will use the common names stabilizers, qubits, checks to refer to elements in $X_L(0), X_L(1), X_L(2)$, respectively.

We will adopt the following notations.
$c_1|_S$ denotes the restriction of the vector $c_1$ to the support on $S$ (equivalently $S \cap X_L(1)$).
$|c_1|_S$ denotes the Hamming weight of the restricted vector $c_1|_S$.

\subsection{Geometric Locality and Bounded Density}

We first show that the codes have the desired geometric properties.
\begin{lemma}
  The codes $Q(X_L)$ are geometrically local and have bounded density.
\end{lemma}
\begin{proof}
  Notice that $\Set(X_L) = V(\cG_L)$
    and the qubit-check relations of the code are precisely the adjacency relations of $\cG_L$.
  This implies that it suffices to check $I_{\sq_L \to \euclid}: V(\cG_L) \to \ZZ^D$ is geometrically local and has bounded density.
  This holds because of \Cref{thm:embedding}.
\end{proof}

\subsection{Dimension}

We next show that the codes retain the same dimensions.
Essentially, what we want to show is that $\cF: X \to X_L$ induces a quasi-isomorphism,
  i.e. $Z^i(X) / B^i(X) \xrightarrow{\sim} Z^i(X_L) / B^i(X_L)$ are isomorphisms for all $i$.
Intuitively, this holds for the same reason that cellular homology of a manifold is independent of the choice of the cellulation.
However, we will not go through those discussions.
We will show $Z^i(X) / B^i(X) \xrightarrow{\sim} Z^i(X_L) / B^i(X_L)$ for $i=1$ directly by brute force.

\begin{lemma}
  If the code $Q(X)$ has dimension $k$ then the code $Q(X_L)$ has dimension $k$.
\end{lemma}
\begin{proof}
  The goal is to show that $\cF_1$ induces a bijection between the codewords (equivalent classes) $Z^1(X)/B^1(X)$ and $Z^1(X_L)/B^1(X_L)$
    by mapping $[\wt c_1]$ to $[\cF_1(\wt c_1)]$.
  To show this, we need to show
  \begin{enumerate}
    \item $\wt c_1 \in B^1(X) \implies \cF_1(\wt c_1) \in B^1(X_L)$.
    \item $\wt c_1 \in Z^1(X) \implies \cF_1(\wt c_1) \in Z^1(X_L)$.
    \item $\cF_1(\wt c_1) \in B^1(X_L) \implies \wt c_1 \in B^1(X)$.
    \item Given $c_1 \in Z^1(X_L)$, there exists $\wt c_1 \in Z^1(X)$, such that $c_1 \in [\cF_1(\wt c_1)]$.
  \end{enumerate}
  The first two imply that the map $[\wt c_1] \mapsto [\cF_1(\wt c_1)]$ is well-defined and maps $Z^1(X)/B^1(X)$ to $Z^1(X_L)/B^1(X_L)$.
  The last two imply that the map is injective and surjective, respectively.
  So, overall, we get a bijection.
  We show the desired statements one by one.

  \paragraph*{1. $\wt c_1 \in B^1(X) \implies \cF_1(\wt c_1) \in B^1(X_L)$.}
  Because $\wt c_1 \in B^1(X)$ there exists $\wt c_0 \in \F2^{X(0)}$ such that $\wt c_1 = \delta_1 \wt c_0$.
  The statement then follows from the left block of the commutative diagram in \Cref{eq:cF}, $\cF_1 (\wt c_1) = \cF_1 (\delta_0 \wt c_0) = \delta_0 \cF_0 (\wt c_0)$.

  \paragraph*{2. $\wt c_1 \in Z^1(X) \implies \cF_1(\wt c_1) \in Z^1(X_L)$.}
  This follows from the right block of the commutative diagram in \Cref{eq:cF}, $\delta_1 \cF_1 (\wt c_1) = \cF_2 (\delta_1 \wt c_1) = 0$.

  \paragraph*{3. $\cF_1(\wt c_1) \in B^1(X_L) \implies \wt c_1 \in B^1(X)$.}
  Let $c_1 = \cF_1(\wt c_1)$ and let $c_0 \in \F2^{X_L(0)}$ be the vector that satisfies $\delta_0 c_0 = c_1$.
  Since $c_1|_S = 0$, $c_0$ takes the same value in each connected component $S_i \in \cS$.
  That means $c_0 \in \im \cF_0$.

  Let $\wt c_0 = \cF_0^{-1} (c_0)$.
  We claim that $\delta_0 c_0 = \wt c_1$.
  From the left block of the commutative diagram in \Cref{eq:cF}, we have $\cF_1 (\delta_0 \wt c_0) = \delta_0 \cF_0 (\wt c_0) = \delta_0 c_0 = c_1 = \cF_1(\wt c_1)$.
  Since $\cF_1$ is injective $\delta_0 c_0 = \wt c_1$.

  \paragraph*{4. Given $c_1 \in Z^1(X_L)$, there exists $\wt c_1 \in Z^1(X)$, such that $c_1 \in [\cF_1(\wt c_1)]$.}
  We first construct $c_1' \in \im \cF_1$ by removing the support of $c_1$ in $S$.
  Notice that the chain complex $X_L|_S: \F2^{S \cap X_L(0)} \to \F2^{S \cap X_L(1)} \to \F2^{S \cap X_L(2)}$ is exact
    because the subcomplex $X_L|_{S_i}$ is exact
      for each connected component $S_i \in \cS$.
  Since $\delta_1 c_1 = 0$, $c_1|_S \in Z^1(X_L|_S) = B^1(X_L|_S)$.
  Therefore, there exists $c_0 \in \F2^{S \cap X_L(0)} = \F2^{X_L(0)}$ such that $(c_1 + \delta_0 c_0)|_S = 0$.
  So $c_1' = c_1 + \delta_0 c_0$ is not supported in $S$ and is only supported in $T$.

  To show $c_1' \in \im \cF_1$,
    we additionally need to show that $c_1'$ takes the same value in each connected component $T_i \in \cT$.
  Notice that $c_1' \in Z^1(X_L)$ which means $c_1'$ violates no checks.
  Each check in $T_i$ acts on two qubits in $T_i$ and some other qubits in $S$.
  Since $c_1'$ is not supported in $S$, the two nearby qubits in $T_i$ take the same value.
  Hence, $c_1'$ take the same value in each connected component.

  Let $\wt c_1 = \cF_1^{-1}(c_1')$.
  We suffice to show $\wt c_1 \in Z^1(X)$.
  From the right block of the commutative diagram in \Cref{eq:cF}, $\cF_2 (\delta_1 \wt c_1) = \delta_1 \cF_1 (\wt c_1) = \delta_1 c_1' = 0$.
  Since $\cF_2$ is injective $\delta_1 \wt c_1 = 0$
    which concludes the proof.
\end{proof}

\begin{corollary}
  $Q(X_L)$ has dimensions $\Omega(n^{\frac{D-2}{D}})$.
\end{corollary}
\begin{proof}
  We first prove the case of $Q(X_L)$ where $n = |X_L(1)|$.
  Recall that \Cref{eq:X-size,eq:XL-size} say
    $|X(1)| = \Theta(L^{D-2})$
    and $|X_L(1)| = \Theta(L^D)$.
  Thus,
  $k(X_L) = k(X) = \Theta(|X(1)|) = \Theta(L^{D-2}) = \Omega(n^{\frac{D-2}{D}})$,
    where the second equality follows from the fact that good codes have linear dimension.
\end{proof}

\subsection{Distance, Energy Barrier: By Showing Small-Set Boundary Expansion}

We now show the rest of the code properties.
As discussed in \Cref{sec:prelim-chain-complex-expansion},
  we will show them by showing the subdivided code has boundary and coboundary expansion.
As a reminder we say a chain complex has $(\alpha,\beta,\gamma)$-coboundary expansion if
\begin{equation*}
  \forall c_1 \in \F2^{X(1)}, \abs{c_1} \le \alpha |X(1)|:
    \exists c_0 \in \F2^{X(0)}, \beta \abs{c_1 + \delta_0 c_0} \le \abs{\delta_1 c_1}, \gamma \abs{c_0} \le \abs{c_1}.
\end{equation*}
\begin{theorem}
  \label{thm:main-expansion}
  If $X$ has $(\alpha_\qLDPC, \beta_\qLDPC, \gamma_\qLDPC)$-small-set coboundary (boundary) expansion,
    then $X_L$ has $(\alpha, \beta, \gamma)$-small-set coboundary (boundary) expansion
    where
  \begin{equation}
    \alpha = \frac{\eta_0^\sur}{4L} \frac{4}{\Delta_{\max}^2}\alpha_{\qLDPC}, \quad
    \beta = \frac{1}{\frac{1}{\beta_1^\sur} + \frac{1}{\beta^\rep \eta_1^\sur} + \frac{1}{\beta_\qLDPC \eta^\rep \eta_1^\sur} \frac{\Delta_{\max}}{2} L}, \quad
    \gamma = \frac{1}{\frac{2}{\beta_0^\sur} + \frac{4}{\gamma_\qLDPC\eta_0^{\sur}} \frac{\Delta_{\max}^2}{4} L}.
  \end{equation}
\end{theorem}
In our context $\beta_0^\sur, \beta_1^\sur, \beta^\rep = \Theta(1/L)$ and $\eta_0^\sur, \eta_1^\sur, \eta^\rep, \alpha_\qLDPC, \beta_\qLDPC, \gamma_\qLDPC, \Delta_{\max} = \Theta(1)$.
Thus, $\alpha, \beta, \gamma = \Theta(1/L)$.

With that we have the following results.
\begin{corollary}
  $Q(X_L)$ has distance $\Omega(n^{\frac{D-1}{D}})$ and energy barrier $\Omega(n^{\frac{D-2}{D}})$.
\end{corollary}
\begin{proof}
  Let $n = |X_L(1)|$.
  As discussed in \Cref{sec:prelim-chain-complex-expansion},
    distance $d \ge \alpha |X_L(1)|$ and energy barrier $\cE \ge \alpha \beta |X_L(1)|$.
  Recall that \Cref{eq:XL-size} say $|X_L(1)| = \Theta(L^D / \polylog L)$.
  Thus,
  $d(X_L) = \Omega(\alpha |X_L(1)|) = \Theta(L^{D-1}) = \Omega(n^{\frac{D-1}{D}})$
  and
  $\cE(X_L) = \Omega(\alpha \beta |X_L(1)|) = \Theta(L^{D-2}) = \Omega(n^{\frac{D-2}{D}})$.
\end{proof}

Before proving the full theorem, we prove the following simpler theorem as a warm up.
\begin{theorem} [A simplification of \Cref{thm:main-expansion}]
  \label{thm:simple-expansion}
  If $X$ has cosystolic (systolic) distance $\alpha_\qLDPC$,
    then $X_L$ has cosystolic (systolic) distance $\alpha = \frac{\eta_0^\sur}{2L} \frac{4}{\Delta_{\max}^2} \alpha_{\qLDPC}$.
  More explicitly,
  \begin{equation}
    \forall c_1 \in Z^1(X_L), |c_1| \le \alpha |X_L(1)| : c_1 \in B^1(X_L).
  \end{equation}
  This means that if $Q(X)$ has linear distance, then $Q(X_L)$ has distance $\Theta(n/L)$ where $n = |X_L(1)|$ is the number of qubits.
\end{theorem}

In the following, we only study the coboundary and cosystolic properties.
Notice that $X$ and $X_L$ are symmetric when flipping the direction of the chain complexes.
So the argument that applies to the coboundary expansion also applies to the boundary expansion.
Therefore, it suffices to study the coboundary expansion.

The proof will utilize the coboundary expansion properties of the generalized surface codes and the generalized repetition codes in \Cref{sec:prop-surface-code}.
We have stated the consequences of coboundary expansion properties explicitly in the proof,
  so it is not required to have a full understanding of \Cref{sec:prop-surface-code} before reading the proof below.
For reference, the definitions and the lemma statements that is used in the proof can be found in \Cref{sec:expansion-properties-generalized-code}.

\begin{proof}
  We first outline the proof.
  To show $c_1 \in B^1(X_L)$,
    the idea is to first reduce the question to another codeword $c_1'$ in the same homology, i.e. $c_1 - c_1' \in B^1(X_L)$.
    (Here, being a codeword means being a systole, i.e. $c_1' \in Z^1(X_L)$.)
  If we can show $c_1' \in B^1(X_L)$, this implies $c_1 \in B^1(X_L)$.

  The codeword $c_1'$ will be constructed in a way so that it is supported in specific 1d-like regions and has a small weight.
    This is achieved using the expansion properties of the generalized surface code.

  Because of the structure of $c_1'$'s support, $c_1'$ can then be associated with another codeword $\wt c_1$ in the old code $X$, which again has a small weight.

  Now, we invoke the distance assumption that small weight codewords of $X$ are boundaries.
  That means $\wt c_1 \in B^1(X)$ which, in turn, implies $c_1' \in B^1(X_L)$, and thus $c_1 \in B^1(X_L)$.

  We now go through each step in detail.

  \paragraph{Step 1: Construct $c_1'$ by removing $c_1$ in $S$, i.e. 2D cleaning.}
  Notice that each connected region of $S$ is a generalized surface code, which has expansion properties by \Cref{cor:sur-expansion}.
  The same applies to their disjoint union with the same expansion parameters.\footnote{In reality, these generalized surface codes are not entirely disjoint. They are disjoint for the interior regions in $S$ but share boundary regions in $T$. Nevertheless, this overlap on the boundary only improves the parameters.}
  By applying \Cref{cor:sur-expansion} to the disjoint regions of $S$ with $\hat f_1 = c_1$, there exist $c^S_0 = f_0$ and $c^S_1 = f_1$ supported on $S$ which satisfy
  \begin{multline}
    (a)\, c_1|_S = (\delta_0 c^S_0 + c^S_1)|_S, \quad
    (b)\, \frac{\beta^\sur_0}{2} |c^S_0|_S \le |c_1|_S, \quad
    (c)\, \beta^\sur_1 |c^S_1|_S \le |\delta_1 c_1|_S, \quad \\
    (d)\, \frac{\eta^\sur_0}{2} |\delta_0 c^S_0|_T \le |c_1|_S, \quad
    (e)\, \eta^\sur_1 |\delta_1 c^S_1|_T \le |\delta_1 c_1|_S.
  \end{multline}
  (The superscript in $c^S_0$ is a reminder that its support is in $S$.
    The subscript is a reminder that $c^S_0$ is a $0$-chain, i.e. $c^S_0 \in \F2^{X_L(0)}$.)
  Notice $c^S_1|_S = 0$, because (c) $\beta_1^\sur |c^S_1|_S \le |\delta_1 \hat c^S_1|_S = 0$.
  We then define $c_1' = c_1 + \delta_0 c^S_0$.
  Notice $c_1'$ has no support in $S$ because (a)
    $c_1'|_S = (c_1 + \delta_0 c^S_0)|_S = c^S_1|_S = 0$.

  Additionally, we have control of the weight $|c_1'|$:
  \begin{align}
    |c_1'|
    &= |c_1'|_T \nonumber\\
    &\le |c_1|_T + |\delta_0 c^S_0|_T \nonumber\\
    &\le |c_1|_T + \frac{2}{\eta^{\sur}_0} |c_1|_S \nonumber\\
    &\le \frac{2}{\eta^{\sur}_0} |c_1| \label{eq:simple-step1}
  \end{align}
  where the third inequality uses (d) and the last inequality uses $1 \le \frac{2}{\eta^{\sur}_0}$.

  \paragraph{Step 2: Construct $c_1''$ by making $c_1'$ consistent in $T$, i.e. 1d cleaning.}
  We say $c_1'$ is consistent on $T$,
    if for each connected region $T_i \subset T$,
    $c_1'$ takes the same value on $T_i$.
  In the current simpler context, $c_1'$ is already consistent,
    because $c_1' \in Z^1$, i.e. $\delta_1 c_1' = 0$, which means $c_1'$ does not violate any check.
  Notice that the check in $T_i$ acts on two qubits in $T_i$ and some other qubits in $S$.
    Given that $c_1'$ is not supported in $S$, the two qubits in $T_i$ must take the same value.
  Therefore, within the connected region $T_i$, all qubits take the same value.

  \paragraph{Step 3: Construct $\wt c_1$ from $c_1''$ by moving from $X_L$ to $X$.}
  The consistency of $c_1'$ and the fact that $c_1'$ is not supported in $S$ implies $c_1' \in \im \cF_1$.
  We define $\wt c_1 = \cF_1^{-1}(c_1')$.

  Notice that the small weight of $c_1'$ implies the small weight of $\wt c_1$.
  Because each region $T_i$ has size at least $|T_i \cap X_L(1)| \ge L$\footnote{From \Cref{fig:proof-regions}, it may seem that we merely have $|T_i \cap X_L(1)| \ge \frac{L+1}{2}$.
  To get $|T_i \cap X_L(1)| \ge L$, we recall that $T_i$ covers at least $2$ branches, since $\Delta_{\min} \ge 2$.
  When there are $2$ branches, precisely $|T_i \cap X_L(1)| = L$.
  Therefore, in the general case, $|T_i \cap X_L(1)| \ge L$.},
    each nonzero value is repeated at least $L$ times.
  This implies $|\wt c_1| \le |c_1'| / L$.

  \paragraph{Step 4: Construct $\wt c_0$ from $\wt c_1$ by applying expansion assumption of $X$.}
  Combine the above, we obtain that the weight $|\wt c_1|$ is small:
  \begin{equation*}
    |\wt c_1| \le |c_1'| / L \le \frac{2}{\eta^\sur_0 L} |c_1|
    \le \frac{2}{\eta^\sur_0 L} \alpha |X_L(1)|
    \le \frac{2L}{\eta^\sur_0} \frac{\Delta_{\max}^2}{4} \alpha |X(1)|
    = \alpha_{\qLDPC} |X(1)|
  \end{equation*}
  where the second inequality uses \Cref{eq:simple-step1} and the fourth inequality uses \Cref{claim:X-bound}, $|X_L(1)| \le \frac{\Delta_{\max}^2}{4} L^2 |X(1)|$.
  Therefore, we can apply the expansion assumption and obtain $\wt c_1 \in B^1(X)$,
    i.e. there exist $\wt c_0 \in \F2^{X(0)}$ such that
    $\wt c_1 = \delta_0 \wt c_0$.

  \paragraph{Step 5 and wrap up.}
  To complete the proof, we need to show $c_1' \in B^1(X_L)$.
  This is done by lifting $\wt c_0 \in \F2^{X(0)}$ into $c_0 = \cF_0(\wt c_0) \in \F2^{X_L(0)}$.
  We have $c_1' = \cF_1(\wt c_1) = \cF_1(\delta_0 \wt c_0) = \delta_0 \cF_0(\wt c_0) = \delta_0 c_0$,
    where the third equality uses the left block of the commutative diagram \Cref{eq:cF}.
\end{proof}

With the proof of the simpler theorem in mind, we are ready to use the same template to show \Cref{thm:main-expansion}.

\begin{proof} [Proof of \Cref{thm:main-expansion}]
  Recall that because $X$ and $X_L$ are symmetric,
    it suffices to study the coboundary expansion which we will now focus on.

  \paragraph{Step 1: Construct $c_1'$ by removing $c_1$ in $S$, i.e. 2D cleaning.}
  By applying \Cref{cor:sur-expansion} to the disjoint regions of $S$ with $\hat f_1 = c_1$, there exist $c^S_0 = f_0$ and $c^S_1 = f_1$ supported on $S$ which satisfy
  \begin{multline} \label{eq:cs0-cs1-dcs0-dcs1}
    (a)\, c_1|_S = (\delta_0 c^S_0 + c^S_1)|_S, \quad
    (b)\, \frac{\beta^\sur_0}{2} |c^S_0|_S \le |c_1|_S, \quad
    (c)\, \beta^\sur_1 |c^S_1|_S \le |\delta_1 c_1|_S, \quad \\
    (d)\, \frac{\eta^\sur_0}{2} |\delta_0 c^S_0|_T \le |c_1|_S, \quad
    (e)\, \eta^\sur_1 |\delta_1 c^S_1|_T \le |\delta_1 c_1|_S.
  \end{multline}
  We then set $c_1'$ to be $c_1 + \delta_0 c^S_0$ restricted to the support on $T \cup U$.
  Notice that $T \cup U$ is the complement of $S$.
  So by (a),
  \begin{equation} \label{eq:c}
    c_1 + \delta_0 c^S_0 = c^S_1 + c_1'.
  \end{equation}
  Intuitively, $c^S_0$ tries to clean up $c_1$ in $S$ using the coboundaries $\delta_0 c^S_0$ as much as possible and
    $c^S_1$ removes the remaining support in $S$.
  The rest of the vector supported in $T \cup U$ is $c_1'$.

  The inequalities (b) and (c) bound the resource needed to remove the support of $c_1$ in $S$,
    which is intuitively saying that one can control the difference between $c_1'$ and $c_1$.
  On the other hand, the inequalities (d) and (e) bound the unwanted effect created in $T \cup U$ caused by the cleaning process,
    which is intuitively saying that $c_1'$ has small weight.
  More concretely, we have the following bounds on the weight
  \begin{align}
    |c_1'| &= |c_1'|_T \nonumber\\
      &\le |\delta_0 c^S_0|_T + |c_1|_T \nonumber\\
      &\le \frac{2}{\eta_0^{\sur}} |c_1|_S + |c_1|_T \nonumber\\
      &\le \frac{2}{\eta_0^{\sur}} |c_1| \label{eq:c1'}
  \end{align}
  where the third inequality uses (d) and the last inequality uses $1 \le \frac{2}{\eta_0^{\sur}}$.
  Furthermore,
  \begin{align}
    |\delta_1 c_1'| &= |\delta_1 c_1'|_T + |\delta_1 c_1'|_U \nonumber\\
      &\le |\delta_1 c^S_1|_T + |\delta_1 c_1|_T + |\delta_1 c_1|_U \nonumber\\
      &\le \frac{1}{\eta_1^{\sur}} |\delta_1 c_1|_S + |\delta_1 c_1|_T + |\delta_1 c_1|_U \nonumber\\
      &\le \frac{1}{\eta_1^{\sur}} |\delta_1 c_1| \label{eq:dc1'}
  \end{align}
  where the second inequality uses $|\delta_1 c_1'|_U = |\delta_1 c_1|_U$
    because by \Cref{eq:c}, $\delta_1 c_1' - \delta_1 c_1 = \delta_1 c^S_1$ which is supported only in $T$.
  The third inequality uses (e) and the last inequality uses $1 \le \frac{1}{\eta_1^{\sur}}$.

  \paragraph{Step 2: Construct $c_1''$ by making $c_1'$ consistent in $T$, i.e. 1d cleaning.}
  By applying \Cref{cor:rep-expansion} to the disjoint regions of $T$
    with $\hat f_0 = c_1'$,
    there exists $c^T_1 = f_0$ supported on $T$ which satisfies
    \begin{equation} \label{eq:ct1-dct1}
      (a)\, c_1' + c^T_1 \in B^1, \quad
      (b)\, |c^T_1|_T \le |c_1'|_T, \quad
      (c)\, \beta^\rep |c^T_1|_T \le |\delta_1 c_1'|_T, \quad
      (d)\, \eta^\rep |\delta_1 c^T_1|_U \le |\delta_1 c_1'|_T.
    \end{equation}
    We then set $c_1''$ to be $c_1' + c^T_1$
    which by (a) violates no checks in $T$, i.e. consistent in $T$.

  We can again bound the weight properties of $c_1''$
  \begin{align} \label{eq:c1''}
    |c_1''| = |c_1'+c^T_1|_T \le |c_1'|_T + |c^T_1|_T \le 2 |c_1'|_T \le \frac{4}{\eta_0^{\sur}} |c_1|
  \end{align}
  where the third inequality uses (b) and the last inequality uses \Cref{eq:c1'}.
  Further,
  \begin{align}
    |\delta_1 c_1''| &= |\delta_1 c_1''|_U \nonumber\\
      &\le |\delta_1 c^T_1|_U + |\delta_1 c_1'|_U \nonumber\\
      &\le \frac{1}{\eta^{\rep}} |\delta_1 c_1'|_T + |\delta_1 c_1'|_U \nonumber\\
      &\le \frac{1}{\eta^{\rep}} |\delta_1 c_1'| \nonumber\\
      &\le \frac{1}{\eta^{\rep} \eta_1^{\sur}} |\delta_1 c_1| \label{eq:dc1''}
  \end{align}
  where the third inequality uses (d) and the last inequality uses \Cref{eq:dc1'}.

  \paragraph{Step 3: Construct $\wt c_1$ from $c_1''$ by moving from $X_L$ to $X$.}
  Because $c_1''$ is consistent on each connected region of $T$,
    we can define $\wt c_1 = \cF_1^{-1} (c_1'') \in \F2^{X(1)}$ as in the proof of \Cref{thm:simple-expansion}.

  Because each connected region of $T$ has size $\ge L$,
    we can again bound the weight properties of $\wt c_1$ where
    $|\wt c_1| \le |c_1''| / L$.
  Furthermore, $|\delta_1 \wt c_1| = |\delta_1 c_1''|$
    because all the checks in $T$ are satisfied
    and by the construction of $\cF_1$,
    $|\delta_1 \wt c_1|_U = |\delta_1 c_1''|_U$.
  (Another way to see this is to use the commutative diagram \Cref{eq:cF} $\delta_1 c_1'' = \delta_1 \cF_1(c_1) = \cF_2(\delta_1 \wt c_1)$.
  Notice that $\cF_2$ does not change the weight of the vector by construction.
  Hence, $|\delta_1 \wt c_1| = |\delta_1 c_1''|$.)

  \paragraph{Step 4: Construct $\wt c_0$ from $\wt c_1$ by applying expansion assumption of $X$.}
  Notice the weight $\wt c_1$ is small:
  \begin{equation*}
    |\wt c_1| \le |c_1''| / L \le \frac{4}{\eta_0^\sur L} |c_1|
    \le \frac{4}{\eta_0^\sur L} \alpha |X_L(1)|
    \le \frac{4L}{\eta_0^\sur} \frac{\Delta_{\max}^2}{4} \alpha |X(1)|
    = \alpha_{\qLDPC} |X(1)|
  \end{equation*}
    where the fourth inequality uses \Cref{claim:X-bound}, $|X_L(1)| \le \frac{\Delta_{\max}^2}{4} L^2 |X(1)|$.
  Therefore, by the expansion assumption of $X$, there exist $\wt c_0 \in X(0)$ such that
    $\beta_{\qLDPC} |\wt c_1 + \delta_0 \wt c_0| \le |\delta_1 \wt c_1|$
    and $\gamma_{\qLDPC} |\wt c_0| \le |\wt c_1|$.

  \paragraph{Step 5: Construct $c_0''$ from $\wt c_0$ by moving back from $X$ to $X_L$.}
  Same as in the proof of \Cref{thm:simple-expansion},
    we lift $\wt c_0$ to $c_0'' = \cF_0(\wt c_0) \in \F2^{X_L(0)}$.

  Notice that each entry of $\wt c_0$ is repeated $\le (\Delta_{\max} \frac{L-1}{2} + 1)^2 \le \frac{\Delta_{\max}^2}{4} L^2$ times,
    and each entries of $\wt c_1$ is repeated between $L$ and $\Delta_{\max} \frac{L-1}{2} + 1 \le \frac{\Delta_{\max}}{2} L$ times.
  Therefore, $\beta_{\qLDPC} |\wt c_1 + \delta_0 \wt c_0| \le |\delta_1 \wt c_1|$ lifts into
    \begin{equation} \label{eq:dc0''}
      \beta_{\qLDPC} |c_1'' + \delta_0 c_0''| \le \frac{\Delta_{\max}}{2} L |\delta_1 c_1''|
    \end{equation}
    because $|c_1'' + \delta_0 c_0''| \le \frac{\Delta_{\max}}{2} L |\wt c_1 + \delta_0 \wt c_0|$
      and $|\delta_1 \wt c_1| = |\delta_1 c_1''|$.
  Further, $\gamma_{\qLDPC} |\wt c_0| \le |\wt c_1|$ lifts into
    \begin{equation} \label{eq:c0''}
      \gamma_{\qLDPC} |c_0''| \le \frac{\Delta_{\max}^2}{4} L |c_1''|
    \end{equation}
    because $|c_0''| \le \frac{\Delta_{\max}^2}{4} L^2 |\wt c_0|$
      and $L |\wt c_1| \le |c_1''|$.

  \paragraph{Wrap up.}
  We are ready to construct $c_0$ with the desired bound $\beta |c_1 + \delta_0 c_0| \le |\delta_1 c_1|$ and $\gamma |c_0| \le |c_1|$.
  This is done by utilizing the vectors we have derived from $c_1$ in the process above, $c_0^S, c_1^S, c_1', c_1^T, c_1'', \wt c_1, \wt c_0, c_0''$, and their associated inequalities.

  We set $c_0 = c_0^S + c_0''$ and check the inequalities:
  \begin{align}
    |c_1 + \delta_0 c_0| &= |c_1 + \delta_0 c_0^S + \delta_0 c_0''| \nonumber\\
      &= |c_1^S + c_1' + \delta_0 c_0''| \nonumber\\
      &= |c_1^S + c_1^T + c_1'' + \delta_0 c_0''| \nonumber\\
      &\le |c_1^S| + |c_1^T| + |c_1'' + \delta_0 c_0''| \nonumber\\
      &\le \frac{1}{\beta_1^\sur} |\delta_1 c_1| + \frac{1}{\beta^\rep} |\delta_1 c_1'|_T + \frac{1}{\beta_\qLDPC} \frac{\Delta_{\max}}{2} L |\delta_1 c_1''| \nonumber\\
      &\le \frac{1}{\beta_1^\sur} |\delta_1 c_1| + \frac{1}{\beta^\rep \eta_1^\sur} |\delta_1 c_1| + \frac{1}{\beta_\qLDPC \eta^\rep \eta_1^\sur} \frac{\Delta_{\max}}{2} L |\delta_1 c_1| \nonumber\\
      &= \left(\frac{1}{\beta_1^\sur} + \frac{1}{\beta^\rep \eta_1^\sur} + \frac{1}{\beta_\qLDPC \eta^\rep \eta_1^\sur} \frac{\Delta_{\max}}{2} L\right) |\delta_1 c_1|
  \end{align}
  where the second and third equality holds by construction.
  The fifth inequality uses \Cref{eq:cs0-cs1-dcs0-dcs1}(c), \Cref{eq:ct1-dct1}(c) and \Cref{eq:dc0''}.
  The sixth inequality uses \Cref{eq:dc1',eq:dc1''}.
  Finally,
  \begin{align}
    |c_0| &= |c_0^S + c_0''| \nonumber\\
      &\le |c_0^S| + |c_0''| \nonumber\\
      &\le \frac{2}{\beta_0^\sur} |c_1|_S + \frac{1}{\gamma_\qLDPC} \frac{\Delta_{\max}^2}{4} L |c_1''| \nonumber\\
      &\le \frac{2}{\beta_0^\sur} |c_1| + \frac{4}{\gamma_\qLDPC\eta_0^{\sur}} \frac{\Delta_{\max}^2}{4} L |c_1| \nonumber\\
      &= \left(\frac{2}{\beta_0^\sur} + \frac{4}{\gamma_\qLDPC\eta_0^{\sur}} \frac{\Delta_{\max}^2}{4} L \right) |c_1|
  \end{align}
  where the third inequality uses \Cref{eq:cs0-cs1-dcs0-dcs1}(b) and \Cref{eq:c0''}.
  The fourth inequality uses \Cref{eq:c1''}.
  This completes the proof.
\end{proof}

\section{Properties of the Generalized Repetition Code and the Generalized Surface Code}
\label{sec:prop-surface-code}

This section focuses on the two \Cref{cor:rep-expansion,cor:sur-expansion} that are used in the proof of \Cref{thm:main-expansion} in the previous section.
They are derived from \Cref{lem:rep-expansion,lem:sur-expansion}
  which are about the coboundary expansion properties of the generalized repetition code and the generalized surface code.

We briefly recall how \Cref{cor:rep-expansion,cor:sur-expansion} were used in the proof of \Cref{thm:main-expansion}.
The proof of \Cref{thm:main-expansion} is to reduce the question of the expansion property of the subdivided code $X_L$
  to the expansion property of the original code $X$.
This reduction is achieved by cleaning up the vector $c_1$ into $c_1'$ then into $c_1''$.
This cleaning is where \Cref{cor:rep-expansion,cor:sur-expansion} comes in.
The property used in cleaning is coboundary expansion and an additional condition that tracks some property on the ``boundary''.

We will discuss this additional condition later in more detail.
For now, we just want to mention that when defining the generalized repetition code and the generalized surface code,
  we additionally consider their ``boundary''.
In particular, the chain complex $\F2^{Y(i)} \xrightarrow{\delta_i} \F2^{Y(i+1)} \xrightarrow{\delta_{i+1}} \F2^{Y(i+2)}$ is extended to have the additional structure $\F2^{Y(i) \cup Y_\bound(i)} \xrightarrow{\delta_i'} \F2^{Y(i+1) \cup Y_\bound(i+1)} \xrightarrow{\delta_{i+1}'} \F2^{Y(i+2) \cup Y_\bound(i+2)}$.\footnote{Note that in the examples we study, $\delta_{i+1}\delta_i = 0$ but $\delta_{i+1}'\delta_i' \ne 0$. The maps $\delta_i'$ do not, and do not need to, form a chain complex.}
We will refer to this structure as a chain complex with boundary.
To simplify the notation, from now on $\delta_i: \F2^{Y(i) \cup Y_\bound(i)} \to \F2^{Y(i+1) \cup Y_\bound(i+1)}$ will mean the map that includes the boundary.
We will use $|f_i|_\int$ to denote the Hamming weight of the vector in $Y$,
  i.e. $|f_i|_\int = |\supp f_i \cap Y(i)|$ (the subscript stands for interior),
  and use $|f_i|_\bound$ to denote the Hamming weight in $Y_\bound$,
  i.e. $|f_i|_\bound = |\supp f_i \cap Y_\bound(i)|$.

\subsection{The Generalized Repetition Code and the Generalized Surface Code}

We first define the generalized repetition code and the generalized surface code,
  then describe their expansion properties.

The repetition code is a 1d chain of bits
  connected by a 1d chain of checks which require the neighboring bits to be the same.
The generalized repetition code is similar,
  except that now there could be one branching point at the center.
See \Cref{fig:generalized-repetition-code}.
The adjacency matrix defines the map $\F2^{Y(0)} \to \F2^{Y(1) \cup Y_\bnd(1)}$.
This can be extended to $Y^\rep: \F2 \xrightarrow{\delta_{-1}} \F2^{Y(0)} \xrightarrow{\delta_0} \F2^{Y(1) \cup Y_\bnd(1)}$
  where $\delta_{-1}$ maps $1$ to the all $1$ vector.
Notice that the restriction of the linear maps to the interior, $\F2 \xrightarrow{\delta_{-1}|_Y} \F2^{Y(0)} \xrightarrow{\delta_0|_Y} \F2^{Y(1)}$, is a chain complex.

We say a generalized repetition code has length $L$ if there are $L$ bits on the path from one boundary to another boundary.
Although the definition of length may seem arbitrary,
  it is defined in a way that eliminates the need for additional conversions when applying the result to the previous section.

\begin{figure}[H]
  \centering
  \includegraphics[width=0.99\linewidth]{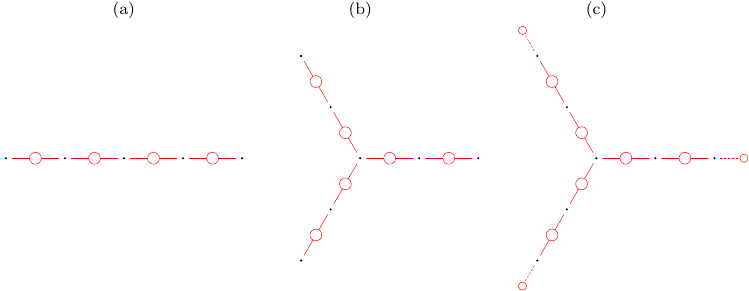}
  \vspace{1em}
  \caption{The figure illustrates the generalized repetition codes of length $5$.
  The black dots are the bits $X(0)$ and the red circles are the checks $X(1)$.
  (a) The standard repetition code.
  (b) The generalized repetition code with degree $3$ branching at the center.
  (c) The same generalized repetition code with boundaries illustrated using the smaller circles.}
  \label{fig:generalized-repetition-code}
\end{figure}

The surface code is the tensor product of two repetition codes.
Similarly, the generalized repetition code is the tensor product of two generalized repetition codes.
See \Cref{fig:generalized-surface-code}.
The adjacency matrices define the map $\F2^{Y(0)} \to \F2^{Y(1) \cup Y_\bnd(1)} \to \F2^{Y(2) \cup Y_\bnd(2)}$
  which again can be extended to
  $Y^\sur: \F2 \xrightarrow{\delta_{-1}} \F2^{Y(0)} \xrightarrow{\delta_0} \F2^{Y(1) \cup Y_\bound(1)} \xrightarrow{\delta_1} \F2^{Y(2) \cup Y_\bound(2)}$
  where $\delta_{-1}$ maps $1$ to the all $1$ vector.
Notice that the restriction of the linear maps on the interior, $\F2 \xrightarrow{\delta_{-1}|_Y} \F2^{Y(0)} \xrightarrow{\delta_0|_Y} \F2^{Y(1)} \xrightarrow{\delta_1|_Y} \F2^{Y(2)}$, is a chain complex.
We say a generalized surface code has length $L$ if the two generalized repetition codes have length $L$.

\begin{figure}[H]
  \centering
  \includegraphics[width=0.99\linewidth]{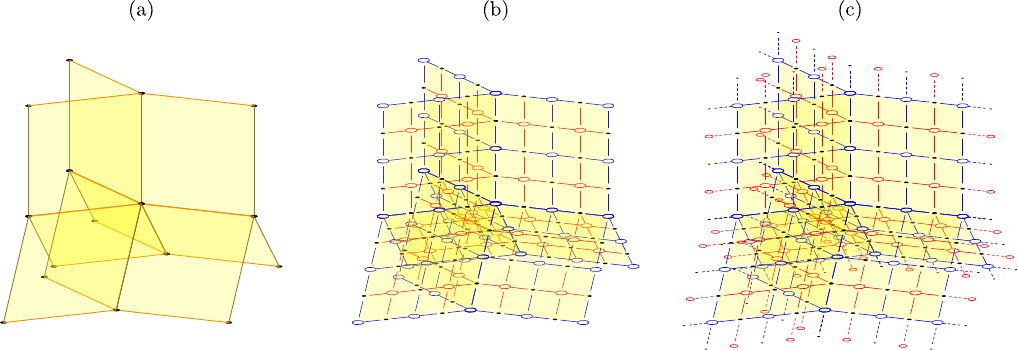}
  \vspace{1em}
  \caption{The figure illustrates the generalized surface codes of length $5$.
  The blue circles are the stabilizers $X(0)$, the black dots are the qubits $X(1)$, and the red circles are the checks $X(2)$.
  (a) The scaffold is the Cartesian product of two Y-shaped structure with degree $3$ branching.
  (b) The generalized surface code.
  (c) The same generalized surface code with boundaries illustrated using the smaller circles.}
  \label{fig:generalized-surface-code}
\end{figure}

\subsection{Expansion Properties of the Generalized Repetition Code and the Generalized Surface Code}
\label{sec:expansion-properties-generalized-code}

Having defined the generalized repetition code and generalized surface code,
  we now describe the desired expansion properties.

\begin{definition} \label{def:coboundary-expansion}
  A chain complex $Y$ with boundary $Y_\bound$ is a $(\beta_i, \eta_i)$-coboundary expander at level $i$
    if for all $\hat f_i \in \F2^{Y(i)}$, there exists $f_i \in \hat f_i + B^i \subset \F2^{Y(i)}$ (recall $B^i = \Ima \delta_{i-1}|_Y$ where both the domain and the range are restricted to $Y$.)
    such that
    \begin{equation} \label{eq:coboundary-expansion}
      (1)\,  |f_i|_\int \le |\hat f_i|_\int, \quad
      (2)\,  \beta_i |f_i|_\int \le |\delta_i \hat f_i|_\int, \quad
      (3)\,  \eta_i |\delta_i f_i|_\bound \le |\delta_i \hat f_i|_\int.
    \end{equation}
\end{definition}

Inequality (2) is the standard coboundary expansion for $Y$ in \Cref{def:standard-coboundary-expansion} and inequality (3) is the extra inequality which bounds the newly introduced objects during cleaning.

The following two lemmas say the generalized repetition code and the generalized surface code have coboundary expansion property.
We will use the subscript to keep track of the level we are discussing.

\begin{lemma} \label{lem:rep-expansion}
  A generalized repetition code with size $L$, $Y^\rep$,
  is a $(\beta_0=\frac{2}{L}, \eta_0=1)$-coboundary expander at level $0$.
\end{lemma}

\begin{lemma} \label{lem:sur-expansion}
  A generalized surface code with size $L$, $Y^\sur$,
  is a $(\beta_0=\frac{1}{L}, \eta_0=\frac{L-1}{4L})$-coboundary expander at level $0$ and a $(\beta_1=\frac{2}{3L}, \eta_1=\frac{1}{2})$-coboundary expander at level $1$.\footnote{The constant prefactors are not tight. Showing $\beta_i = \Theta(1/L)$ and $\eta_i = \Theta(1)$ is sufficient for our purpose.}
\end{lemma}

We attempt to phrase the lemmas above to elucidate the structure.
However, it is not obvious how the current statements relate to the main proof presented in the previous section.
So the following rephrasing seeks to enhance explicitness and usability.

\begin{corollary} \label{cor:rep-expansion}
  Given a generalized repetition code with size $L$, $Y^\rep: \F2 \xrightarrow{\delta_{-1}} \F2^{Y(0)} \xrightarrow{\delta_0} \F2^{Y(1) \cup Y_\bound(1)}$, for all $\hat f_0 \in \F2^{Y(0)}$, there exists $f_0 \in \F2^{Y(0)}$ such that
  \begin{equation} \label{eq:rep-expansion}
    (a)\, \hat f_0 + f_0 \in B^0, \quad
    (b)\, |f_0|_\inter \le |\hat f_0|_\inter, \quad
    (c)\, \beta^\rep |f_0|_\inter \le |\delta_0 \hat f_0|_\inter, \quad
    (d)\, \eta^\rep |\delta_0 f_0|_\bound \le |\delta_0 \hat f_0|_\inter,
  \end{equation}
  with $\beta^\rep = \frac{2}{L}$ and $\eta^\rep = 1$.
\end{corollary}

\begin{proof} [Proof of \Cref{cor:rep-expansion}]
  By \Cref{lem:rep-expansion}, we can find $f_0 \in \hat f_0 + B^0$ where $f_0$ satisfies the inequalities in \Cref{eq:coboundary-expansion} at level $i=0$.
  By construction, we have (a).
  One easily check that (1) $\implies$ (b), (2) $\implies$ (c), (3) $\implies$ (d).
\end{proof}

\begin{corollary} \label{cor:sur-expansion}
  Given a generalized surface code with size $L$, $Y^\sur: \F2 \xrightarrow{\delta_{-1}} \F2^{Y(0)} \xrightarrow{\delta_0} \F2^{Y(1) \cup Y_\bound(1)} \xrightarrow{\delta_1} \F2^{Y(2) \cup Y_\bound(2)}$, for all $\hat f_1 \in \F2^{Y(1)}$, there exist $f_0 \in \F2^{Y(0)}$ and $f_1 \in \F2^{Y(1)}$ such that
  \begin{multline} \label{eq:sur-expansion}
    (a)\, \hat f_1 = (\delta_0 f_0 + f_1)|_\inter, \quad
    (b)\, \frac{\beta^\sur_0}{2} |f_0|_\inter \le |\hat f_1|_\inter, \quad
    (c)\, \beta^\sur_1 |f_1|_\inter \le |\delta_1 \hat f_1|_\inter, \quad \\
    (d)\, \frac{\eta^\sur_0}{2} |\delta_0 f_0|_\bound \le |\hat f_1|_\inter, \quad
    (e)\, \eta^\sur_1 |\delta_1 f_1|_\bound \le |\delta_1 \hat f_1|_\inter,
  \end{multline}
  with $\beta^\sur_0 = \frac{1}{L}, \eta^\sur_0 = \frac{L-1}{4L}, \beta^\sur_1 = \frac{2}{3L}, \eta^\sur_0 = \frac{1}{2}$.
\end{corollary}

\begin{proof} [Proof of \Cref{cor:sur-expansion}]
  By \Cref{lem:sur-expansion}, we can write $\hat f_1 = f_1 + (\delta_0 \hat f_0)|_\int$ for some $\hat f_0 \in \F2^{X(0)}$ where $f_1$ satisfies the inequalities in \Cref{eq:coboundary-expansion} at level $i=1$.
    And we can find $f_0 \in \hat f_0 + B^0$ where $f_0$ satisfies the inequalities in \Cref{eq:coboundary-expansion} at level $i=0$.
  By construction, we have (a).
  One can check that (2) at level 1 $\implies$ (c) and
  (3) at level 1 $\implies$ (e).
  To show (b) and (d), we use (2) and (3) at level 0,
    $\beta^\sur_0 |f_0|_\inter \le |\delta_0 \hat f_0|_\inter$ and $\eta^\sur_0 |\delta_0 f_0|_\bnd \le |\delta_0 \hat f_0|_\inter$
    together with
  \begin{equation}
    |\delta_0 \hat f_0|_\inter \le |\hat f_1|_\inter + |f_1|_\inter \le 2 |\hat f_1|_\inter
  \end{equation}
  where the last inequality uses (1) at level 1, $|f_1|_\inter \le |\hat f_1|_\inter$.
\end{proof}

The remainder of this section is devoted to showing \Cref{lem:rep-expansion,lem:sur-expansion}.

\subsection{Proof of \Cref{lem:rep-expansion} for Generalized Repetition Codes}
\label{sec:proof-lem-rep-expansion}

We observe that there are only two elements in $\hat f_0 + B^0$ since $B^0 = \set{0, \mathbbm{1}}$ where $0$ is the all $0$ vector and $\mathbbm{1}$ is the all $1$ vector.
So all is left is to figure out which choice satisfies the inequalities.
We denote $\Delta$ as the degree of the branching.
We call the branching point the root.

\begin{proof} [Proof of \Cref{lem:rep-expansion}]
  We set $f_0 = \hat f_0$ if $|\hat f_0|_\int \le |Y(0)|/2$.
  Otherwise, we set $f_0 = \hat f_0 + \mathbbm{1}$.
  This immediately satisfies inequality (1) $|f_0|_\int \le |\hat f_0|_\int$.
  Notice that $|f_0|_\int \le |Y(0)|/2$.

  We now show (2) $\beta_0 |f_0|_\int \le |\delta_0 \hat f_0|_\int$ with $\beta_0 = \frac{2}{L}$.
    If $|\delta_0 \hat f_0|_\int \ge \Delta/2$, because $|f_0|_\int \le |Y(0)|/2$ and $|Y(0)| = \Delta \frac{L-1}{2} + 1$, the inequality is satisfied.
  Otherwise, $|\delta_0 \hat f_0|_\int < \Delta/2$, this means $\delta_0 \hat f_0|_\int$ is not supported on $\Delta - |\delta_0 \hat f_0|_\int$ branches
    which means these branches take the same value as the root.
    We observe that this value has to be $0$
      otherwise the root and at least half of the branches have value $1$
      which violates $|f_0|_\int \le |Y(0)|/2$.
  Therefore, $|f_0|_\int \le \frac{L}{2} |\delta_0 \hat f_0|_\int$,
    since $f_0$ can be supported on at most $|\delta_0 \hat f_0|_\int$ branches and each branch has weight at most $\frac{L}{2}$.

  Finally, we show (3) $\eta_0 |\delta_0 f_0|_\bnd \le |\delta_0 \hat f_0|_\int$ with $\eta_0 = 1$.
    We divide the $\Delta$ branches into different groups based on the number of $1$s of $\delta_0 \hat f_0|_\int$ in each branch.
    Let $s, t, u$ be the number of branches with no $1$, with an odd number of $1$s, and an even nonzero number of $1$s, respectively.
    Then $|\delta_0 \hat f_0|_\int \ge t + 2u$.
  If the root value of $f_0$ is $0$,
    then $|\delta_0 f_0|_\bnd = t$ which satisfies the inequality
    $|\delta_0 f_0|_\bnd \le |\delta_0 \hat f_0|_\int$.
  Otherwise, the root value of $f_0$ is $1$,
    then $|\delta_0 f_0|_\bnd = s + u$.
  So it suffices to show $s \le t + u$.
  Notice these $s$ branches takes the same value as the root which is $1$.
  Since $|f_0|_\int \le |Y(0)|/2$, this means $s$ covers at most half of the $\Delta$ branches.
  Thus, $s \le t + u$.
\end{proof}

\subsection{Proof of \Cref{lem:sur-expansion} for Generalized Surface Codes}

One might hope to demonstrate the expansion of the generalized surface code using two facts: (1) the generalized surface code is a tensor product of two generalized repetition codes (2) the expansion property is preserved under tensor product.
It is indeed the case that we can establish the level $0$ expansion through this approach, which we will explain.
However, it is not clear to us if a tensor product argument applies to the level $1$ expansion.
Therfore, we will show the level $1$ expansion in an ad hoc method through case studies.
It would be appealing to have a simpler proof.

\subsubsection{Level $0$ Expansion}
Note that the level $0$ expansion is essentially the isoperimetric inequalities which is known to behave nicely under tensor product \cite{chung1998isoperimetric,tillich2000edge}.
This gives inequality (2).
To show inequality (3), we will apply a similar proof technique.

Consider the graph $\cG_1$ corresponding to the generalized repetition code where $Y^\rep(0)$ are the vertices and $Y^\rep(1)$ (without $Y^\rep_\bnd(1)$) are the edges.
Similarly, consider the graph $\cG$ corresponding to the generalized surface code where $Y^\sur(0)$ are the vertices and $Y^\sur(1)$ (without $Y^\sur_\bnd(1)$) are the edges.
We observe that the graph from the generalized surface code $\cG$ is the Cartesian product of two graphs $\cG_1$ and $\cG_2$ originating from the generalized repetition code.
To capture the notion of boundary $|\delta_0 f_0|_\bnd$ in the context of graphs,
  we also consider the boundary vertices which are the vertices in $Y^\rep$ ($Y^\sur$) next to $Y^\rep_\bnd$ ($Y^\sur_\bnd$).
  See \Cref{fig:proof-level-0-boundary}.

\begin{definition} [Cartesian product of two graphs with boundary]
  Given two (undirected) graphs $\cG_1$ and $\cG_2$.
  Let $V_i$, $E_i$, and $V_i^\bnd \subset V_i$ be the vertex set, the edge set, and the boundary vertex set of $\cG_i$, respectively.

  The Cartesian product of two graphs $\cG_1$ and $\cG_2$, denoted as $\cG_1 \times \cG_2$, is a graph with vertex set $V = V_1 \times V_2$
  and two vertices $(x_1, x_2)$ and $(y_1, y_2)$ are adjacent iff $x_2 = y_2$ and $\{x_1, y_1\}$ is an edge of $\cG_1$ or $x_1 = y_1$ and $\{x_2, y_2\}$ is an edge of $\cG_2$ which can be expressed as $E = (E_1 \times V_2) \cup (V_1 \times E_2)$.
  The new boundary vertex (multi)set is $V^\bnd = (V_1^\bnd \times V_2) \sqcup (V_1 \times V_2^\bnd)$, where $\sqcup$ is the disjoint union.
\end{definition}

\begin{figure}[H]
  \centering
  \includegraphics[width=0.7\linewidth]{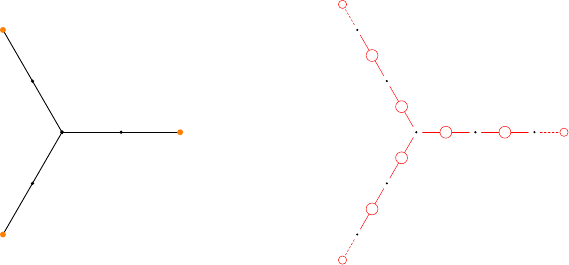}
  \vspace{1em}
  \caption{An illustration of $\cG_1$ on the left that corresponds to $Y^\rep$ on the right. Note that $V_1^\bnd$ is marked in orange.}
  \label{fig:proof-level-0-boundary}
\end{figure}

We interpret $V^\bnd$ as a multiset which double counts the vertices $V_1^\bnd \times V_2^\bnd$
  so that $|\delta_0 f_0|_\bnd$ for the generalized surface code can be expressed as $\sum_{x \in V^\bnd} 1_{f_0(x) \ne 0}$.

It is known that isoperimetric inequalities are closely related to functional inequalities.
These functional inequalities behave nicely under tensor product.
Therefore, our strategy is to use the functional inequalities of the generalized repetition codes
  to derive the functional inequalities for the generalized surface codes,
  which will then give the desired isoperimetric inequalities.

\begin{definition}
  We say a graph with boundary $\cG=(V, E, V^\bnd)$ satisfies $(C, C^\bnd)$-functional inequalities if for all function $h: V \to \set{0,1}$\footnote{The functional inequalities are typically studied for more general functions $h: V \to \RR$.
    In general, these two scenarios are equivalent under some weak assumptions.
    However, we do not need this stronger form and will omit this discussion.
    For more information, see \cite[Sec 2]{tillich2000edge}.}
  \begin{gather}
    \sum_{\{x, y\} \in E} |h(x) - h(y)| \ge \frac{C}{|V|} \sum_{x, y \in V} |h(x) - h(y)|, \label{eq:functional-inequality-1} \\
    \sum_{\{x, y\} \in E} |h(x) - h(y)| \ge \frac{C^\bnd}{|V^\bnd|} \sum_{x \in V^\bnd, y \in V} |h(x) - h(y)|. \label{eq:functional-inequality-2}
  \end{gather}
\end{definition}

\begin{lemma}
  \label{lem:functional-inequality}
  If $\cG_1, \cG_2$ satisfy $(C, C^\bnd)$-functional inequalities,
  then $\cG_1 \times \cG_2$ satisfies $(C, \min(C, C^\bnd))$-functional inequalities.
\end{lemma}

We will prove the lemma at the end of this section.
We first discuss how the functional inequalities and the coboundary expander are related.
We then combine everything to prove the level $0$ coboundary expansion of the surface code.

Let $h$ be $f_0$ where $f_0$ is the vector with the smaller weight between $\hat f_0$ and $\hat f_0 + \mathbbm{1}$.
In the language of chain complexes, the two functional inequalities become
\begin{gather}
  |\delta_0 f_0|_\int \ge \frac{C}{|V|} 2|f_0|\Big(|V|-|f_0|\Big), \label{eq:functional-inequality-f0-1} \\
  |\delta_0 f_0|_\int \ge \frac{C^\bnd}{|V^\bnd|} \bigg(|\delta_0 f_0|_\bnd \Big(|V|-|f_0|\Big) + \Big(|V^\bnd|-|\delta_0 f_0|_\bnd\Big) |f_0|\bigg). \label{eq:functional-inequality-f0-2}
\end{gather}
Applying this relation, we obtain the following facts.

\begin{claim} \label{claim:repetition-code-has-functional-inequality}
  The graphs corresponding to the generalized repetition code $\cG_1, \cG_2$ satisfy $(C=\frac{1}{L}, C^\bnd=\frac{1}{L})$-functional inequalities.
\end{claim}

\begin{proof}
  Recall the generalized repetition code has $(\beta_0 = \frac{2}{L}, \eta_0=1)$-coboundary expansion.
  In particular, for $|f_0| \le |V|/2$
  \begin{gather}
    |\delta_0 f_0|_\int \ge \frac{2}{L} |f_0|, \\
    |\delta_0 f_0|_\int \ge |\delta_0 f_0|_\bnd.
  \end{gather}

  Since $|V| \ge |V| - |f_0|$ and $|V^\bnd| \ge |V^\bnd| - |\delta_0 f_0|_\bnd$, the inequalities above imply
  \begin{equation}
    |\delta_0 f_0|_\int
      \ge \frac{1/L}{|V|} 2 |f_0| |V|
      \ge \frac{1/L}{|V|} 2 |f_0| \Big(|V|-|f_0|\Big),
  \end{equation}
  \begin{align}
    \Big(\frac{L}{2} + \frac{L}{2}\Big) |\delta_0 f_0|_\int
      &\ge \frac{L}{2} |\delta_0 f_0|_\bnd + |f_0| \\
      &\ge \frac{1}{|V^\bnd|} \Big(|\delta_0 f_0|_\bnd |V| + |V^\bnd| |f_0|\Big) \\
      &\ge \frac{1}{|V^\bnd|} \bigg(|\delta_0 f_0|_\bnd \Big(|V|-|f_0|\Big) + \Big(|V^\bnd|-|\delta_0 f_0|_\bnd\Big) |f_0|\bigg)
  \end{align}
  where the second inequality uses $\frac{L}{2} \ge \frac{|V|}{|V^\bnd|}$.

  Thus, by comparing to \Cref{eq:functional-inequality-f0-1,eq:functional-inequality-f0-2}, the functional inequalities are satisfied for $(C=\frac{1}{L}, C^\bnd=\frac{1}{L})$.
\end{proof}

\begin{claim} \label{claim:functional-inequality-implies-coboundary-expansion}
  If the graph with boundary satisfies $(C, C^\bnd)$-functional inequalities,
    then the corresponding complex is a $(\beta_0 = C, \eta_0 = \frac{C^\bnd|V|}{2|V^\bnd|})$-coboundary expander at level $0$.
\end{claim}

\begin{proof}
  Let $f_0$ be the vector with the smaller weight between $\hat f_0$ and $\hat f_0 + \mathbbm{1}$.
  Since $|f_0| \le |V|/2$, \Cref{eq:functional-inequality-f0-1,eq:functional-inequality-f0-2} imply
  \begin{gather}
    |\delta_0 f_0|_\int \ge C |f_0|, \\
    |\delta_0 f_0|_\int \ge \frac{C^\bnd|V|}{2|V^\bnd|} |\delta_0 f_0|_\bnd.
  \end{gather}
\end{proof}

We now combine the above to prove the level $0$ coboundary expansion of the surface code.
\begin{proof} [Proof of \Cref{lem:sur-expansion} at level $0$]
  By \Cref{claim:repetition-code-has-functional-inequality},
    $\cG_1, \cG_2$ satisfy $(C=\frac{1}{L}, C^\bnd=\frac{1}{L})$-functional inequalities.
  Together with \Cref{lem:functional-inequality},
    this implies $\cG$ satisfies $(\frac{1}{L}, \frac{1}{L})$-functional inequalities.
  Together with \Cref{claim:functional-inequality-implies-coboundary-expansion},
    this implies the generalized surface code satisfies $(\beta_0=\frac{1}{L}, \eta_0=\frac{L-1}{4L})$-coboundary expansion,
    where we use $\frac{|V|}{|V^\bnd|} = \frac{|V_1||V_2|}{|V_1^\bnd||V_2| + |V_1||V_2^\bnd|} \ge \frac{L-1}{4}$ which holds because $\frac{|V_1|}{|V_1^\bnd|}, \frac{|V_2|}{|V_2^\bnd|} \ge \frac{L-1}{2}$.
\end{proof}

We now go back and prove the lemma.
\begin{proof} [Proof of \Cref{lem:functional-inequality}]
  We first review the proof of \Cref{eq:functional-inequality-1} in \cite{tillich2000edge},
    then extend the proof to \Cref{eq:functional-inequality-2}.
  The idea is to first split the edges into those running horizontally and those running vertically.
  We then apply the functional inequalities for $\cG_1$ and $\cG_2$.
  Finally, we use the triangle inequality to rewrite the terms.
  \begin{align*}
    \sum_{\{x, y\} \in E} |f(x) - f(y)|
    &= \sum_{\{x_1, y_1\} \in E_1, x_2 \in V_2} |f(x_1, x_2) - f(y_1, x_2)|
      + \sum_{y_1 \in V_1, \{x_2, y_2\} \in E_2} |f(y_1, x_2) - f(y_1, y_2)| \\
    &\ge \frac{C}{|V_1|} \sum_{x_1, y_1 \in V_1, x_2 \in V_2} |f(x_1, x_2) - f(y_1, x_2)|
      + \frac{C}{|V_2|} \sum_{y_1 \in V_1, x_2, y_2 \in V_2} |f(y_1, x_2) - f(y_1, y_2)| \\
    &\ge \frac{C}{|V_1||V_2|} \sum_{x_1, y_1 \in V_1, x_2, y_2 \in V_2} |f(x_1, x_2) - f(y_1, y_2)| \\
    &= \frac{C}{|V|} \sum_{x, y \in V} |f(x) - f(y)|
  \end{align*}

  We now apply the same technique to \Cref{eq:functional-inequality-2}, but this time, we do it in reverse order.
  \begin{align*}
    \sum_{x \in V^\bnd, y \in V} |f(x) - f(y)|
    &= \sum_{x \in V_1^\bnd \times V_2, y \in V} |f(x) - f(y)|
      + \sum_{x \in V_1 \times V_2^\bnd, y \in V} |f(x) - f(y)| \\
    &\hspace*{-6.4em}\le |V_2| \sum_{x_1 \in V_1^\bnd, y_1 \in V_1, x_2 \in V_2} |f(x_1, x_2) - f(y_1, x_2)|
      + |V_1^\bnd| \sum_{y_1 \in V_1, x_2, y_2 \in V_2} |f(y_1, x_2) - f(y_1, y_2)| \\
      &\hspace*{-5.4em}+ |V_1| \sum_{x_1 \in V_1, x_2 \in V_2^\bnd, y_2 \in V_2} |f(x_1, x_2) - f(x_1, y_2)|
      + |V_2^\bnd| \sum_{x_1, y_1 \in V_1, y_2 \in V_2} |f(x_1, y_2) - f(y_1, y_2)| \\
    &\hspace*{-6.4em}\le \frac{|V_1^\bnd||V_2|}{C^\bnd} \sum_{\{x_1, y_1\} \in E_1, x_2 \in V_2} |f(x_1, x_2) - f(y_1, x_2)|
      + \frac{|V_1^\bnd||V_2|}{C} \sum_{y_1 \in V_1, \{x_2, y_2\} \in E_2} |f(y_1, x_2) - f(y_1, y_2)| \\
      &\hspace*{-5.4em}+ \frac{|V_1||V_2^\bnd|}{C^\bnd} \sum_{x_1 \in V_1, \{x_2, y_2\} \in E_2} |f(x_1, x_2) - f(x_1, y_2)|
      + \frac{|V_1||V_2^\bnd|}{C} \sum_{\{x_1, y_1\} \in E_1, y_2 \in V_2} |f(x_1, y_2) - f(y_1, y_2)| \\
    &\hspace*{-6.4em}\le \frac{|V^\bnd|}{\min(C^\bnd, C)} \sum_{\{x, y\} \in E} |f(x) - f(y)|.
  \end{align*}
  Notice $|V^\bnd| = |V_1^\bnd||V_2| + |V_1||V_2^\bnd|$.
\end{proof}

\begin{remark}
  There are other functional inequalities that behave nicely under tensor product.
  This includes the Poincar\'e and log Sobolev inequalities
    (see Proposition 4.3.1 and Proposition 5.2.7 in \cite{bakry2014analysis})
    and hypercontractivity of the boolean hypercube
    (see \cite[Proof of the Bonami Lemma]{o2014analysis}).
\end{remark}

\subsubsection{Level $1$ Expansion}

As discussed, the level $1$ expansion is proven in an ad hoc manner.

\begin{proof} [Proof of \Cref{lem:sur-expansion} at level $1$]
  The idea is to decompose the support of $\hat f_1$ into components.
  (For now, one can think of these as the connected components.)
  Because of the triangle inequality,
    it is sufficient to establish the desired bound, (1), (2), and (3) in \Cref{eq:coboundary-expansion}, for each component individually.
  It is straightforward to show the desired bound for the component that is supported on a flat region.
  The remaining challenge is to show the desired bound for the components that crosses the seam which will be shown at the end.
  We refer to the lines at intersection of the planes as the seam.

  \begin{figure}[H]
    \centering
    \includegraphics[width=0.3\linewidth]{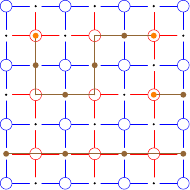}
    \vspace{1em}
    \caption{The figure illustrate an example of $\hat f_1$. The brown dots are the support of $\hat f_1$.
    The orange dots are the violated checks, i.e. the support of $\delta_1 \hat f_1|_\int$.
    The graph consist of edges connecting the support of $\hat f_1$ with the nearby checks.
    We see that the graph decomposes into disjoint $\hat f_{1,i}$.}
    \label{fig:surface-code-decompose-path}
  \end{figure}

  \paragraph{Step 1: Decompose $\hat f_1$ into components.}
  The goal of this step is to write $\hat f_1 = \sum_i \hat f_{1,i}$
    where $\supp \hat f_1$ is a disjoint union of $\supp \hat f_{1,i}$.
  Additionally, the support of $\hat f_{1,i}$ will look like 1d paths which will be made precise later.

  To obtain the decomposition,
    we consider the (bipartite) graph consisting of edges formed by the active qubits in $\supp \hat f_1$ and their nearby checks.
    We call a check an active check if it is next to an active qubit.
  We see that a non-violated check has an even degree while a violated check has an odd degree.
    See \Cref{fig:surface-code-decompose-path}.
  Based on this structure of the graph, we can decompose the graph into components
    whose end points are the violated checks or boundary qubits
    and whose active check has degree at most two.
  Note that when a check has more than two neighboring active qubits, there is more than one possible decomposition satisfying the above properties.
    In such cases, any such decomposition works.
  Notice that when the component does not contain qubits on the seam,
    the component is either a 1d loop or a 1d path
    which are simple objects to analyze directly.

  To be more mathematical, we can write $\hat f_1 = \sum_i \hat f_{1,i}$
    where $\set{\hat f_{1,i}}_i$ (the active qubits) have disjoint supports
    and $\set{\delta_1 \hat f_{1,i}|_\int}_i$ (the violated checks) have disjoint supports.
  Furthermore, for each $\hat f_{1,i}$:
  \begin{itemize}
    \item The graph formed by the qubits and their neighboring checks, $\{(v, w) \in X(1) \times X(2): v \in \supp \hat f_{1,i}, v \text{ is adjacent to } w\}$, is connected.
    \item Each check $w$ in the graph has degree at most two.
  \end{itemize}
  An example of the decomposition is illustrated in \Cref{fig:surface-code-decompose-path}(a).

  Note that to find $f_1$,
    it suffices to find $f_{1,i}$ for each $\hat f_{1,i}$.
  This is because once given $f_{1,i}$, we can set $f_1 = \sum_i f_{1,i}$.
  By construction, we have $f_1 \in \hat f_1 + B^1$ because $f_{1,i} \in \hat f_{1,i} + B^1$.
  The inequalities also carry over to $f_i$.
  For example, to show inequality (1) in \Cref{eq:coboundary-expansion}, we have
  \begin{equation}
    |f_1| \le \sum_i |f_{1,i}| \le \sum_i |\hat f_{1,i}| = |\hat f_1|
  \end{equation}
  where the first inequality holds by the triangle inequality,
    the second inequality, $|f_{1,i}| \le |\hat f_{1,i}|$, is based on the reduced problem for $\hat f_{1,i}$,
    and the third inequality holds because $\set{\hat f_{1,i}}_i$ have disjoint supports.
  Inequalities (2) and (3) in \Cref{eq:coboundary-expansion} can be obtained similarly using the fact that $\set{\delta_1 \hat f_{1,i}|_\int}_i$ have disjoint supports.

  \paragraph{Step 2: Bound each component.}
  We now find $f_{1,i}$ for four different cases of $\hat f_{1,i}$.
  The last case is the hardest.

  \paragraph*{Case 1: The cluster has no violated check, i.e. $\delta_1 \hat f_{1,i} = 0$.}
    In this case, one can simply set $f_{1,i} = 0$ which satisfies (1), (2), and (3).

  \paragraph*{Case 2: The cluster is not connected to the boundary nor the seam (excluding Case 1).}
    The corresponding graph of $\hat f_{1,i}$ is a path that connects two violated checks.
    We can set $f_{1,i}$ to be a shortest path between the two violated checks,
      for example, the path that goes vertically then horizontally.

  We now check (1), (2), and (3).
  Inequality (1) is satisfied because $f_{1,i}$ is the vector with the smallest weight in $\hat f_{1,i} + B^1$.
  Inequality (2) holds because $|f_{1,i}| \le L$ and $|\delta_1 \hat f_{1,i}|_\int = 2$.
  Inequality (3) holds because $|\delta_1 f_{1,i}|_\bnd = 0$.
  (In particular, the inequalities hold whenever $\beta_1 \le \frac{2}{L}$.)

  \paragraph*{Case 3: The cluster is connected to the boundary but not the seam (excluding Case 1).}
    The corresponding graph of $\hat f_{1,i}$ is a path that connects one violated check and the boundary.
    We can set $f_{1,i}$ to be a shortest path between the violated check and the boundary which either goes up vertically or right horizontally.

  We now check (1), (2), and (3).
  Inequality (1) is satisfied because $f_{1,i}$ is the vector with the smallest weight in $\hat f_{1,i} + B^1$.
  Inequality (2) holds because $|f_{1,i}| \le L/2$ and $|\delta_1 \hat f_{1,i}|_\int = 1$.
  Inequality (3) holds because $|\delta_1' f_{1,i}|_\bnd = |\delta_1 \hat f_{1,i}|_\int = 1$
  (In particular, the inequalities hold whenever $\beta_1 \le \frac{2}{L}, \eta_1 \le 1$.)

  \paragraph*{Case 4: Otherwise, the cluster is connected to the seam.}
    The corresponding graph of $\hat f_{1,i}$ is more complicated, so we cannot apply structure results directly as above.
    What we do here is to reduce $\hat f_{1,i}$ to another vector $\hat f_{1,i}'$ that is supported next to the seam.
      (The region next to the seam is depicted in \Cref{fig:proof-level-1-next-to-seam}.)
    This is done by cleaning up the surface region using ideas similar to those in Case 1, 2, 3,
      and pushing the support next to but not past the seam.
    Vector $\hat f_{1,i}'$ will then be analyzed in the next step.

  It is straightforward to check that the structure of the graph in the 2D area is either:
  \begin{itemize}
    \item A path that connects two qubits on the seam.
    \item A path that connects one qubit on the seam and a violated check.
    \item A path that connects one qubit on the seam and the boundary.
  \end{itemize}
  This allows us to decompose $\hat f_{1,i}$ into $\hat f_{1,i} = \hat f_{1,i}^x + \sum_j \hat f_{1,i,j}^y + \sum_k \hat f_{1,i,k}^z + \sum_l \hat f_{1,i,l}^w$ with disjoint support where
  \begin{itemize}
    \item $\hat f_{1,i}^x$ is supported on the seam.
    \item $\hat f_{1,i,j}^y$ are the parts that connect two qubits on the seam.
    \item $\hat f_{1,i,k}^z$ are the parts that connect one qubit on the seam to the boundary.
    \item $\hat f_{1,i,l}^w$ are the parts that connect one qubit on the seam to a violated check.
  \end{itemize}
  Furthermore, $\set{|\delta_1 f_{1,i,k}^z|_\bnd}_k$ are disjoint and $\set{|\delta_1 f_{1,i,l}^w|_\int}_l$ are disjoint.

  We now simplify $\hat f_{1,i,j}^y, \hat f_{1,i,k}^z, \hat f_{1,i,l}^w$ into $f_{1,i,j}^y, f_{1,i,k}^z, f_{1,i,l}^w$ within the same homology, i.e. $g_1 \in \hat g_1 + B^1$ for $g_1$ being $f_{1,i,j}^y$, $f_{1,i,k}^z$, or $f_{1,i,l}^w$.
  Additionally, we have bound on weights $|g_1| \le |\hat g_1|$.
  \begin{itemize}
    \item For $\hat f_{1,i,j}^y$, we simplify it into $f_{1,i,j}^y$ which connects the same two qubit through the route next to the seam.
    \item For $\hat f_{1,i,k}^z$, we simplify it into $f_{1,i,k}^z$ which is the shortest path between the qubit and the boundary. This path is next to the seam.
    \item For $\hat f_{1,i,l}^w$, we simplify it into $f_{1,i,l}^w$ which is the shortest path between the qubit and the violated check.
  \end{itemize}

  Let $\hat f_{1,i}' = \hat f_{1,i}^x + \sum_j f_{1,i,j}^y + \sum_k f_{1,i,k}^z$.
  By construction $\hat f_{1,i}'$ is supported next to the seam.

  We claim that we can reduce the question of finding $f_{1,i}$ to the question of finding $f_{1,i}'$ for $\hat f_{1,i}'$ which will satisfy the inequalities
  \begin{equation} \label{eq:f-next-to-seam}
    (1)\,  |f_{1,i}'|_\int \le |\hat f_{1,i}'|_\int, \quad
    (2)\,  \frac{2}{L} |f_{1,i}'|_\int \le |\delta_1 \hat f_{1,i}'|_\int, \quad
    (3)\,  \frac{1}{2} |\delta_1 f_{1,i}'|_\bound \le |\delta_1 \hat f_{1,i}'|_\int.
  \end{equation}
  The reason is that suppose we have $f_{1,i}' \in \hat f_{1,i}' + B^1$,
    we can then set $f_{1,i} = f_{1,i}' + \sum_l f_{1,i,l}^w$.
  One can straightforwardly check that $f_{1,i} \in \hat f_{1,i} + B^1$.

  We now check inequalities (1), (2), and (3) in \Cref{eq:coboundary-expansion}.
  Inequality (1) holds because
  \begin{align*}
    |f_{1,i}|
    &\le |f_{1,i}'| + \sum_l |f_{1,i,l}^w| \\
    &\le |\hat f_{1,i}'| + \sum_l |f_{1,i,l}^w| \\
    &\le |\hat f_{1,i}^x| + \sum_j |f_{1,i,j}^y|  + \sum_k |f_{1,i,k}^z| + \sum_l |f_{1,i,l}^w| \\
    &\le |\hat f_{1,i}^x| + \sum_j |\hat f_{1,i,j}^y| + \sum_k |\hat f_{1,i,k}^z| + \sum_l |\hat f_{1,i,l}^w| \\
    &= |\hat f_{1,i}|
  \end{align*}
  where the last equality holds because the components are disjoint.

  Inequality (2) holds for $\beta_1 = \frac{2}{3L}$ because
  \begin{align*}
    |f_{1,i}|
    \le |f_{1,i}'| + \sum_l |f_{1,i,l}^w|
    \le \frac{L}{2} |\delta_1 \hat f_{1,i}'|_\int + \sum_l L
    \le \left(\frac{L}{2} + L\right) |\delta_1 \hat f_{1,i}|_\int.
  \end{align*}
  The second inequality uses $|f_{1,i,l}^w| \le L$.
  The third inequality uses $|\delta_1 \hat f_{1,i}'|_\int = |\delta_1 \hat f_{1,i}|_\int$
    and $\sum_l 1 = |\delta_1 \hat f_{1,i}|_\int$.

  Inequality (3) holds for $\eta_1 = \frac{1}{2}$ because
  \begin{align*}
    |\delta_1 f_{1,i}|_\bnd = |\delta_1 f_{1,i}'|_\bnd \le 2 |\delta_1 \hat f_{1,i}|_\int = 2 |\delta_1 \hat f_{1,i}|_\int.
  \end{align*}

  The remaining task is to find $f_{1,i}'$ for $\hat f_{1,i}'$ that is supported next to the seam.

  \begin{figure}[H]
    \centering
    \includegraphics[width=0.4\linewidth]{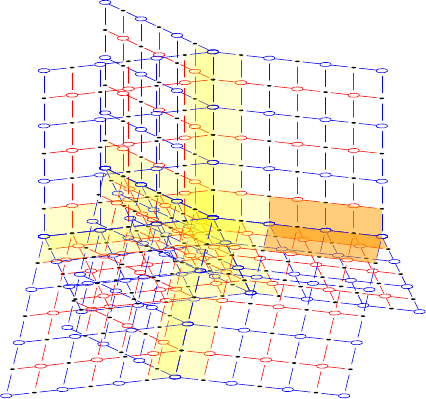}
    \vspace{1em}
    \caption{The figure illustrate the regions next to the seam in yellow.
    Each orange region will be referred to as a limb.}
    \label{fig:proof-level-1-next-to-seam}
  \end{figure}

  \paragraph{Step 3: Show the case where $\hat f_{1,i}'$ is next to the seam.}

  We simply pick $f_{1,i}'$ to be the vector with the smallest weight in $\hat f_{1,i}' + B^1$ while being next to the seam.
  This automatically satisfies inequality (1) in \Cref{eq:f-next-to-seam}.

  For inequality (2) in \Cref{eq:f-next-to-seam}, we construct a vector $f_{1,i}'' \in \hat f_{1,i}' + B^1$
    with $\frac{2}{L} |f_{1,i}''|_\int \le |\delta_1 \hat f_{1,i}'|_\int$.
    Since $f_{1,i}'$ has the smallest weight, it implies inequality (2).
  One can simply construct $f_{1,i}''$ by connecting each violated check, i.e. $\supp \delta_1 \hat f_{1,i}'$ to the closest boundary.
  Since each point is at most $\frac{L-1}{2}$ distance away from the boundary, we have $|f_{1,i}''|_\int \le \frac{L-1}{2} |\delta_1 \hat f_{1,i}'|_\int \le \frac{L}{2} |\delta_1 \hat f_{1,i}'|_\int$.
  (We need the tighter bound $|f_{1,i}''|_\int \le \frac{L-1}{2} |\delta_1 \hat f_{1,i}'|_\int$ in the next paragraph.)

  For inequality (3) in \Cref{eq:f-next-to-seam}, to show $\frac{1}{2} |\delta_1 f_{1,i}'|_\bnd \le |\delta_1 \hat f_{1,i}'|_\int$, we use proof by contradiction.
  We assume $|\delta_1 f_{1,i}'|_\bnd > 2 |\delta_1 \hat f_{1,i}'|_\int$ and show that $f_{1,i}'$ is not the vector in $\hat f_{1,i}' + B^1$ with the smallest weight.
  In particular, we will show $|f_{1,i}'| > \frac{L-1}{2} |\delta_1 \hat f_{1,i}'|_\int$.
  Because $|f_{1,i}''| \le \frac{L-1}{2} |\delta_1 \hat f_{1,i}'|_\int$ from the last paragraph,
    $f_{1,i}'$ has a larger weight than $f_{1,i}''$ which is the contradiction.

  We first perform some basic structural analysis of $f_{1,i}'$
    then observe a local condition of $f_{1,i}'$.
  Consider $v \in \supp f_{1,i}'$ on a horizontal seam.
  We study the corresponding limb as illustrated in \Cref{fig:proof-level-1-next-to-seam}.
    The edges emanating from $v$ either end at the neighboring violated check, or go towards the center, or go towards the boundary.
  Let the number of edges of each type be $p, q, r$, respectively.
  We claim that $p, q, r$ satisfies $p + q \ge r$.
  Otherwise, if $p + q < r$, and one can replace $f_{1,i}'$ with $f_{1,i}' + \delta_0 1_w$,
    where $w$ is the stabilizer next to the qubit $v$ that is closer to the boundary.
  See \Cref{fig:proof-level-1-pqr}(a)(b) for an illustration.

  \begin{figure}[H]
    \centering
    \includegraphics[width=0.99\linewidth]{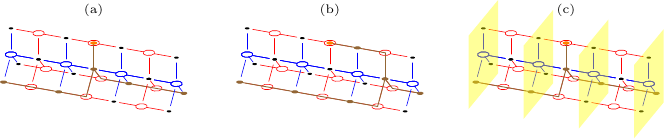}
    \caption{The figure illustrates one of the limbs where left is the direction to the center and right is the direction to the boundary.
    (a) The minimal vector $f_{1,i}'$. (b) The new vector $f_{1,i}' + \delta_0 1_w$ which has a lower weight if $p + q < r$. (c) The cross sections.}
    \label{fig:proof-level-1-pqr}
  \end{figure}

  We now use the local condition $p + q \ge r$ and the assumption $|\delta_1 f_{1,i}'|_\bnd > 2 |\delta_1 \hat f_{1,i}'|_\int$ to show $|f_{1,i}'| > \frac{L-1}{2} |\delta_1 \hat f_{1,i}'|_\int$.
  Consider one of the limbs.
  Suppose the limb contains $s$ violated checks and $t$ lines that connect to the boundary.
  Consider the number of active qubits on each cross section
    as shown in \Cref{fig:proof-level-1-pqr}(c).
  We claim that this number is at least $t-s$ for each cross section.
  Assuming so, since there are $\frac{L-1}{2}$ cross sections,
    this limb contains at least $\frac{L-1}{2} (t-s)$ active qubits,
    i.e. $|\supp f_{1,i}' \cap M| \ge \frac{L-1}{2} (t-s)$
    where $M$ is the set of vertices on the limb.
  We apply this observation to all limbs.
  Because the total number of violated checks is at most $|\delta_1 \hat f_{1,i}'|_\int$ and the number of lines that connect to the boundary is $|\delta_1 f_{1,i}'|_\bnd$,
    this means the number of active qubits $|f_{1,i}'|$ is at least $\frac{L-1}{2} (|\delta_1 f_{1,i}'|_\bnd - |\delta_1 \hat f_{1,i}'|_\int) > \frac{L-1}{2} |\delta_1 \hat f_{1,i}'|_\int$.
  It suffices to show the claim.

  Notice the number of active qubits on the cross section
    changes only when there is an active qubit on the seam
    between the two cross sections.
  With the same notation $p, q, r$,
    this number changes from $r$ to $q$ as we go closer to the center.
  The number drops by at most $p$ because the local condition says $q \ge r - p$.
  Since the number of active qubits is $t$ on the cross section closest to boundary and the number drops by at most $s$,
    each cross section has at least $t-s$ active qubits.
\end{proof}

\section{Application to Classical Codes} \label{sec:classical-LDPC-code}

The main part of the paper discusses the construction based on quantum LDPC codes.
This section explores the implications when applying the techniques to classical LDPC codes
  and shows \Cref{thm:classical-LTC}.
Parts of the result in section are discovered independently in \cite{baspin2023combinatorial}.

\subsection{Code Construction}

Given a classical LTC \cite{panteleev2021asymptotically,dinur2021locally,lin2022c,leverrier2022quantum,dinur2022good}, let $\F2^{X(0)} \to \F2^{X(1)}$ be the parity-check matrix where $X(0)$ are the bits and $X(1)$ are the checks.
The code naturally corresponds to the adjacency matrix of a bipartite graph $\cG$ with vertices $V(\cG) = X(0) \cup X(1)$.
We can subdivide both the code and the graph as shown in \Cref{fig:subdivide-3}.
We denote $X_L$ as the subdivided code and $\cG_L$ as the subdivided graph.

\begin{figure}[H]
  \centering
  \includegraphics[width=0.62\textwidth]{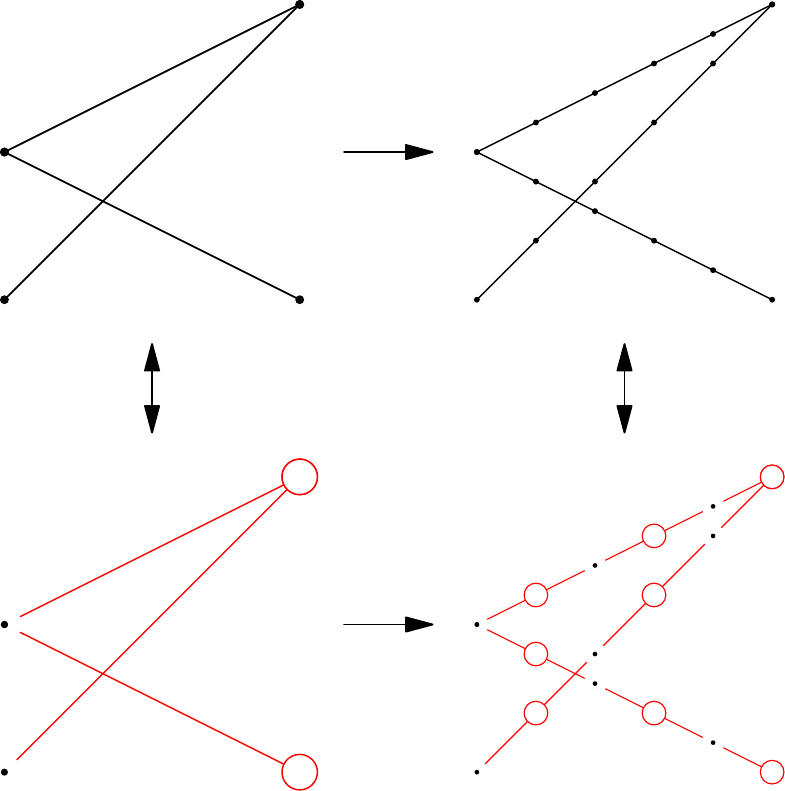}
  \caption{The subdivision of a classical LTC $X$ together with the subdivision of the associated bipartite graph $\cG$.}
  \label{fig:subdivide-3}
\end{figure}

The embedding of the code again boils down to the embedding of the subdivided graph similar to \Cref{thm:embedding}.
\begin{theorem} [Upcoming]
  \label{thm:embedding-1d}
  Given a bounded degree graph $\cG$,
  the $L$-subdivision $\cG_L$ has an embedding $f: V(\cG_L) \to \ZZ^D$ ($D\ge 2$) with constant $a, b = \Theta(1)$ such that
  \begin{enumerate}
    \item Geometrically local: For all adjacent vertices on the complex $(v_0, v_1) \in E(\cG_L)$,
      the distance of corresponding points in $\ZZ^D$ is bounded $|f(v_0) - f(v_1)| \le a$.
    \item Bounded density: The number of vertices at each point is bounded,
      i.e. for all $x \in \ZZ^D$, $|f^{-1}(x)| \le b$.
  \end{enumerate}
  for $L = \Theta(|V(\cG)|^{\frac{1}{D-1}})$
    where the constant depends only on the degree of the complex and the dimension $D$.
\end{theorem}

\begin{remark}
  When $D = 3$ (or $D \ge 3$ with an additional generalization), this embedding theorem is similar to the work by Kolmogorov and Barzdin \cite{barzdin1993realization}.\footnote{We thank Elia Portnoy for pointing out the work and its generalization.}
  The slight difference is that Kolmogorov and Barzdin
    study the notion of thick embedding,
    which prohibits the edges of the embedded graph from overlapping,
    whereas in our case, we allow overlaps as long as they remain to be of bounded density.
  This difference allows our embedding theorem to work even when $D = 2$.
\end{remark}

\subsection{Code Properties}

It is clear that the code has a geometrically local embedding induced from the graph. So we will mainly focus on the non geometric properties, including code dimension, distance, energy barrier, and soundness.

Similar to \Cref{sec:subcomplex}, we define the chain map $\cF$, which maps the original LTC $\F2^{X(0)} \to \F2^{X(1)}$ to the subdivided LTC $\F2^{X_L(0)} \to \F2^{X_L(1)}$.
This utilizes the bijection between vertices and connected components depicted in \Cref{fig:proof-regions-1d}
  where $\cT$ and $\cU$ are the sets of connected components of $T$ and $U$.
\begin{itemize}
  \item Given $\wt c_0 \in \F2^{X(0)}$, we define $\cF_0(\wt c_0) \in \F2^{X_L(0)}$
    by repeating the value $\wt c_0(v)$ at the corresponding component $T_i$ of $v$ for each $v \in X(0)$.
  \item Given $\wt c_1 \subset \F2^{X(1)}$, we define $\cF_1(\wt c_1) \in \F2^{X_L(1)}$
    by setting the value $\wt c_1(v)$ at the corresponding vertex $U_i$ of $v$ for each $v \in X(1)$
    and set other values in $T$ to be $0$.
\end{itemize}
The maps form the following commutative diagram.
\begin{equation}\label{eq:cF-classical}
  \begin{tikzcd}
    \F2^{X(0)} \arrow[r, "\delta_0"] \arrow[d, "\cF_0"] &
    \F2^{X(1)} \arrow[d, "\cF_1"] \\
    \F2^{X_L(0)} \arrow[r, "\delta_0"] & \F2^{X_L(1)}
  \end{tikzcd}
\end{equation}

\begin{figure}[H]
  \centering
  \includegraphics[width=0.99\textwidth]{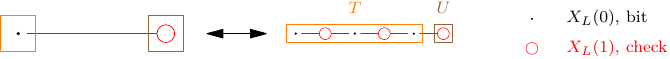}
  \caption{Regions $T, U$ and the bijections $X(0) \cong \cT$, $X(1) \cong \cU$.}
  \label{fig:proof-regions-1d}
\end{figure}

The following lemma relates the code properties of the original code $C(X)$ and the subdivided code $C(X_L)$.
\begin{lemma}
  If the code $C(X)$ has dimension $k_{\LTC}$ then the code $C(X_L)$ has dimension $k = k_{\LTC}$.

  If the code $C(X)$ has distance $d_{\LTC}$ then the code $C(X_L)$ has distance $d \ge L\, d_{\LTC}$.

  If the code $C(X)$ has soundness $s_{\LTC}$ then the code $C(X_L)$ has soundness
  \begin{equation}
    s \ge \frac{|X_L(0)|}{|X_L(1)|} \frac{1}{\frac{1}{\beta^\rep} + \frac{1}{\eta^\rep} \frac{|X(0)|}{|X(1)|} \frac{1}{s_{\LTC}} \frac{\Delta_{\max}}{2} L}.
  \end{equation}
\end{lemma}
In our context $\beta^\rep = \Theta(1/L)$, $\eta^\rep, \Delta_{\max} = \Theta(1)$, $|X_L(0)| = \Theta(L |X(0)|)$, and $|X_L(1)| = \Theta(L |X(1)|)$.
Therefore, when $X$ is chosen to be the LTC with linear dimension, linear distance, constant soundness, and satisfies $|X(1)| = \Theta(|X(0)|)$ provided in \cite{panteleev2021quantum,dinur2021locally,lin2022c,leverrier2022quantum,dinur2022good}, $X_L$ gives the desired properties \Cref{thm:classical-LTC}.
(Recall the bound on energy barrier can be inferred from the bound on distance and soundness as discussed in \Cref{sec:prelim-codes}.)

The lemma can be derived using similar arguments as those presented in \Cref{sec:prop-subdivided-code}.
Here, we provide only a proof sketch for code dimension and distance.
The code dimension of $C(X_L)$ can be estabilished by showing
  $\cF_0$ induces a bijection between the codewords of $C(X)$ and $C(X_L)$.
The distance of $C(X_L)$ can be shown by using this bijection and the fact that all bits in a component $T_i$ take the same value where $T_i$ has size at least $L$.

We now discuss the proof of soundness.
\begin{proof} [Proof of soundness]
  The goal is to show that for each $c_0 \in \F2^{X_L(0)}$,
  there exists a codeword $h_0 \in Z^0(X_L)$ such that
  \begin{equation} \label{eq:proof-soundness-1}
    \frac{s}{|X_L(0)|} |c_0 - h_0| \le \frac{1}{|X_L(1)|} |\delta_0 c_0|.
  \end{equation}

  The idea is to apply the coboundary expansion of the generalized repetition code to replace $c_0$ with $c_0'$ which is in the image of $\cF_0$.
  We then move back to the original code with $\wt c_0 = \cF_0^{-1} (c_0') \in \F2^{X(0)}$.
  This allows us to apply the soundness property of $X$ to find a codeword $\wt h_0 \in Z^0(X)$.
  The codeword $\wt h_0$ then induces the codeword $h_0 = \cF_0(\wt h_0) \in Z^0(X_L)$ which we claim to satisfy the desired inequaliy, \Cref{eq:proof-soundness-1}.

  \paragraph*{Step 1: Construct $c_0'$ by making $c_0$ consistent in $T$.}
  By applying \Cref{cor:rep-expansion} to the disjoint regions of $T$ with $\hat f_0 = c_0$, there exist $c_0^T = f_0$ supported on $T$ which satisfies
  \begin{equation} \label{eq:ct0-dct0}
    (a)\, c_0 + c^T_0 \in B^0, \quad
    (b)\, |c^T_0|_T \le |c_0|_T, \quad
    (c)\, \beta^\rep |c^T_0|_T \le |\delta_0 c_0|_T, \quad
    (d)\, \eta^\rep |\delta_0 c^T_0|_U \le |\delta_0 c_0|_T.
  \end{equation}
  We set $c_0' = c_0 + c^T_0$ which by (a) violates no check in $T$, i.e. consistent in $T$.

  Notice
  \begin{equation} \label{eq:dc0'}
    |\delta_0 c_0'|
    = |\delta_0 c_0'|_U
    \le |\delta_0 c_0^T|_U + |\delta_0 c_0|_U
    \le \frac{1}{\eta^\rep} |\delta_0 c_0|_T + |\delta_0 c_0|_U
    \le \frac{1}{\eta^\rep} |\delta_0 c_0|
  \end{equation}
  where the third inequality uses (d) and the last inequality uses $1 \le \frac{1}{\eta^\rep}$.

  \paragraph*{Step 2: Construct $\wt c_0$ from $c_0'$ by moving from $X_L$ to $X$.}
  The consistency of $c_0'$ implies $c_0' \in \im \cF_0$.
  Therefore, we define $\wt c_0 = \cF_0^{-1} c_0'$.

  Notice
  \begin{equation} \label{eq:dtc0}
    |\delta_0 c_0'| = |\delta_0 \cF_0 (\wt c_0)| = |\cF_1 (\delta_0 \wt c_0)| = |\delta_0 \wt c_0|
  \end{equation}
  where the third equality holds because, by construction, $\cF_1$ preserves the weight.

  \paragraph*{Step 3: Construct $\wt h_0$ from $\wt c_0$ by applying soundness assumption of $X$.}
  By soundness of $C(X)$, there exists $\wt h_0 \in Z^0(X)$ such that
  \begin{equation} \label{eq:X-LTC}
    \frac{s_\LTC}{|X(0)|} |\wt c_0 - \wt h_0| \le \frac{1}{|X(1)|} |\delta_0 \wt c_0|.
  \end{equation}

  \paragraph*{Step 4: Construct $h_0$ from $\wt h_0$ by moving back from $X$ to $X_L$.}
  Let $h_0 = \cF_0 (\wt h_0)$. By the commutative diagram \Cref{eq:cF-classical}, this implies $h_0 \in Z^0(X_L)$ is a codeword.

  Notice that
  \begin{equation} \label{eq:c0'h0}
    |c_0' - h_0| = |\cF_0(\wt c_0 - \wt h_0)| \le \frac{\Delta_{\max}}{2} L |\wt c_0 - \wt h_0|
  \end{equation}
  because each entry in $\wt c_0 - \wt h_0$ is repeated at most $\Delta_{\max} \frac{L-1}{2} + 1 \le \frac{\Delta_{\max}}{2} L$ times.

  \paragraph*{Wrap up.}
  Finally, we show the inequality, \Cref{eq:proof-soundness-1}.
  \begin{align}
    |c_0 - h_0|
    &= |-c_0^T + c_0' - h_0| \nonumber\\
    &\le |c_0^T| + |c_0' - h_0| \nonumber\\
    &\le \frac{1}{\beta^\rep} |\delta_0 c_0| + \frac{\Delta_{\max}}{2} L |\wt c_0 - \wt h_0| \nonumber\\
    &\le \frac{1}{\beta^\rep} |\delta_0 c_0| + \frac{|X(0)|}{|X(1)|} \frac{1}{s_{\LTC}} \frac{\Delta_{\max}}{2} L |\delta_0 \wt c_0| \nonumber\\
    &= \frac{1}{\beta^\rep} |\delta_0 c_0| + \frac{|X(0)|}{|X(1)|} \frac{1}{s_{\LTC}} \frac{\Delta_{\max}}{2} L |\delta_0 c_0'| \nonumber\\
    &\le \frac{1}{\beta^\rep} |\delta_0 c_0| + \frac{1}{\eta^\rep} \frac{|X(0)|}{|X(1)|} \frac{1}{s_{\LTC}} \frac{\Delta_{\max}}{2} L |\delta_0 c_0| \nonumber\\
    &= \left(\frac{1}{\beta^\rep} + \frac{1}{\eta^\rep} \frac{|X(0)|}{|X(1)|} \frac{1}{s_{\LTC}} \frac{\Delta_{\max}}{2} L\right) |\delta_0 c_0|
  \end{align}
  where the third inequality uses \Cref{eq:ct0-dct0}(c) and \Cref{eq:c0'h0}.
  The fourth inequality uses \Cref{eq:X-LTC}.
  The fifth equality uses \Cref{eq:dtc0}.
  The sixth inequality uses \Cref{eq:dc0'}.
\end{proof}

\section*{Acknowledgements}
We thank Jeongwan Haah for pointing out the paper on welding construction and helpful discussions on his lattice quantum codes.
We thank Elia Portnoy for discussions on the mathematical developments related to quantum codes and high dimensional expanders.
We thank John McGreevy and Tarun Grover for suggesting possible applications to physics and the SYK model.
TCL was supported in part by funds provided by the U.S. Department of Energy (D.O.E.) under the cooperative research agreement DE-SC0009919 and by the Simons Collaboration on Ultra-Quantum Matter, which is a grant from the Simons Foundation (652264 JM).

\bibliographystyle{unsrt}
\bibliography{references.bib}

\appendix

\section{Upper Bounds on Dimension, Distance, and Energy Barrier}
\label{sec:upper-bound}

This section collects various known upper bound on the geometrically local codes which our constructions have matched up to $\polylog(n)$.
We assume the codes have at most $1$ qubit (bit) at each lattice site of $[L]^D$ where $[L] = \{0, 1, 2, ..., L-1\}$.

We first discuss the bound on distance and energy barrier of the geometrically local quantum codes.
\begin{theorem}[{\cite[Theorem 1, 2]{bravyi2009no}}]
  Suppose $L \ge 2(r-1)^2$.
  For geometrically local quantum codes in $[L]^D$
  \begin{equation}
    d \le r L^{D-1},\quad \cE \le 6 r^2 L^{D-2}
  \end{equation}
  where each stabilizer generator is bounded by a hypercube of size $r^D$.
\end{theorem}
Note that in \cite[Theorem 2]{bravyi2009no}, although the authors did not state the case where the dimension is greater than $2$, the proof technique applies directly and yield the bound stated above.

The idea is that one can show the existence of a logical operator in a slab of width $r$, which implies $d \le r L^{D-1}$.
(The slab has width $r$ in one of the directions and length $L$ in the other $D-1$ directions)
Then one can create this logical operator layer by layer (in one of the length $L$ directions).
Only the end layer costs energy which gives $\cE \le 6 r^2 L^{D-2}$.

The same idea applies to classical codes.
\begin{theorem}
  For geometrically local classical codes in $[L]^D$
  \begin{equation}
    d \le L^D,\quad \cE \le r L^{D-1}
  \end{equation}
  where each check is bounded by a hypercube of size $r^D$.
\end{theorem}

We now discuss the tradeoff between code dimension and distance.
\begin{theorem}[{\cite[Equation 4, 5]{bravyi2010tradeoffs}}]
  For quantum codes
  \begin{equation}
    k \le \frac{2 D r^2 L^D}{\left(\frac{d}{2Dr}\right)^{\frac{2}{D-1}}} = \Theta\bigg(\frac{L^D}{d^{\frac{2}{D-1}}}\bigg)
  \end{equation}
  for large enough $d$.

  For classical codes
  \begin{equation}
    k \le L^D \left(1 - \frac{d}{\left(d^{\frac{1}{D}} + r\right)^D}\right) \le \frac{D r L^D}{d^{\frac{1}{D}}} = \Theta\bigg(\frac{L^D}{d^{\frac{1}{D}}}\bigg)
  \end{equation}
  for large enough $d$ where $d \ge r^D$.
\end{theorem}

We leave it as an open question about the tradeoff between code dimension and the energy barrier.

\end{document}